\pdfoutput=1

\newif\ifFull
\Fullfalse

\newif\ifAppendix
\Appendixfalse

\newif\ifLipics
\Lipicsfalse

\newif\ifSketch
\Sketchtrue

\newif\ifAbbv
\Abbvtrue

\newif\ifNewSchool
\NewSchoolfalse

\newif\ifRemove
\Removetrue

\ifFull
\documentclass[letterpaper]{journal}
\newcommand\sketch[0]{Sketch}
\else
\ifLipics
\documentclass[a4paper,USenglish]{lipics-v2016}
\usepackage[activate={true,nocompatibility},final,tracking=true,kerning=true,spacing=true]{microtype}
\newcommand\sketch[0]{Proof (Sketch)}
\else
\documentclass[runningheads]{llncs}
\newcommand\sketch[0]{Sketch}
\fi 
\fi

\usepackage{xspace}
\usepackage{graphicx}
\usepackage {subfig}

\ifFull
\usepackage{color}
\usepackage[dvipsnames,table]{xcolor}
\usepackage{amssymb,amsfonts,amsmath,amsthm}
\definecolor {sepia} {rgb} {0.75,0.30,0.15}
\everymath{\color{sepia}}
\usepackage[margin=1in] {geometry}
\usepackage{algorithm}
\usepackage{algorithmic}
\usepackage[colorlinks,urlcolor=NavyBlue,linkcolor=JungleGreen,citecolor=JungleGreen]{hyperref}
\usepackage{subfigure}

\newtheorem{definition}{Definition}
\newtheorem{theorem}{Theorem}
\newtheorem{corollary}{Corollary}
\newtheorem{lemma}{Lemma}
\newtheorem{claim}{Claim}
\else
\ifLipics
\newtheorem{observation}{Observation}

\else
\usepackage{color}
\usepackage{hyperref}
\newtheorem{observation}{Observation}
\newtheorem{cclaim}{Claim}

\let\doendproof\endproof
\renewcommand\endproof{~\hfill$\boxtimes$\doendproof}

\fi 
\fi 

\usepackage{hyperref}
\usepackage{amssymb,amsfonts,amsmath}
\usepackage{enumerate}

\definecolor {infocolor} {rgb} {0.6,0.6,0.6}
\definecolor {ideacolor} {rgb} {0.2,0.6,0.2}

\newenvironment {repeatlemma} [1]
{\noindent {\bf Lemma~\ref{#1}.}\ \slshape} {\normalfont}

\title{Reconstructing Generalized Staircase Polygons with Uniform Step Length}

\ifFull
\author{Nodari Sitchinava\thanks{Department of Information and Computer Sciences, University of Hawaii at Manoa, USA.
  \texttt{nodari@hawaii.edu}}
\and Darren Strash\thanks{
Department of Computer Science, Colgate University, Hamilton, NY, USA.
  \texttt{dstrash@cs.colgate.edu}}%
\else
\ifLipics
\titlerunning{Staircase Polygons with Uniform Step Length} 
\author[1]{Nodari Sitchinava\thanks{This material is based upon work supported by the National Science Foundation under Grant No. 1533823.}}
\author[2]{Darren Strash}
\affil[2]{Department of Information and Computer Sciences, University of Hawaii at Manoa, USA.\\
  \texttt{nodari@hawaii.edu}}
\affil[2]{Department of Computer Science, Colgate University, Hamilton, NY, USA.\\
  \texttt{dstrash@cs.colgate.edu}}

\authorrunning{N. Sitchinava and D. Strash}

\Copyright{Nodari Sitchinava and Darren Strash}

\subjclass{F.2.2 Nonnumerical Algorithms and Problems, G.2.2 Graph Theory}
\keywords{Dummy keyword -- please provide 1--5 keywords}

\else
\author{Nodari Sitchinava\inst{1}$^\text{,}$\thanks{This material is based upon work supported by the National Science Foundation under Grant No. 1533823.}, \and Darren Strash\inst{2}}
\institute{Department of Information and Computer Sciences, University of Hawaii at Manoa, USA.
  \email{nodari@hawaii.edu}%
  \and%
Department of Computer Science, Colgate University, Hamilton, NY, USA.
  \email{dstrash@cs.colgate.edu}%
}

\titlerunning{Reconstructing Generalized Staircase Polygons with Uniform Step Length}
\authorrunning{N. Sitchinava and D. Strash}

\fi 
\fi 

\ifLipics
\EventEditors{John Q. Open and Joan R. Acces}
\EventNoEds{2}
\EventLongTitle{42nd Conference on Very Important Topics (CVIT 2016)}
\EventShortTitle{CVIT 2016}
\EventAcronym{CVIT}
\EventYear{2016}
\EventDate{December 24--27, 2016}
\EventLocation{Little Whinging, United Kingdom}
\EventLogo{}
\SeriesVolume{42}
\ArticleNo{23}
\fi

\newcommand{\etal}{{et al.}}

\newcommand{\niceremark}[3]{\textcolor{red}{\textsc{#1 #2: }}\textcolor{blue}{\textsf{#3}}}
\newcommand{\nodari}[2][says]{\niceremark{Nodari}{#1}{#2}}
\newcommand{\darren}[2][says]{\niceremark{Darren}{#1}{#2}}
\newcommand{\idea}[1]{\textbf{Idea:}~\textcolor{ideacolor}{#1}}

\renewcommand{\idea}[1]{}	      
\renewcommand{\niceremark}[3]{}  

\newcommand{\C}{\mathcal{C}}
\newcommand{\R}{\mathcal{R}}
\newcommand{\fig}{Figure}
\renewcommand{\sec}{Section}
\ifFull
\else
\ifLipics
\else
\renewcommand{\fig}{Fig.}
\renewcommand{\sec}{Sect.}
\fi 
\fi 

\begin{document}
\maketitle

\begin{abstract}
\emph{Visibility graph reconstruction}, which asks us to construct a polygon that has a given visibility graph, is a fundamental problem with unknown complexity (although visibility graph recognition is known to be in PSPACE).
We show that two classes of uniform step length polygons can be reconstructed efficiently by finding and removing rectangles formed between consecutive convex boundary vertices called tabs. 
In particular, we give an $O(n^2m)$-time reconstruction algorithm for orthogonally convex polygons, 
where $n$ and $m$ are the number of vertices and edges in the visibility graph, respectively.
We further show that reconstructing a monotone chain of staircases (a histogram) is fixed-parameter tractable, when parameterized on the number of tabs, and polynomially solvable in time $O(n^2m)$ under reasonable alignment restrictions.

\keywords{Visibility graphs $\cdot$ Polygon reconstruction $\cdot$ Visibility graph recognition $\cdot$ Orthogonal polygons $\cdot$ Fixed-parameter tractability}
\end{abstract}

\section {Introduction}
Visibility graphs, used to capture visibility in or between polygons, are simple but powerful tools in computational geometry. They are integral to solving many fundamental problems, such as routing in polygons, and art gallery and watchman problems, to name a few.
Efficient, and even worst-case optimal, algorithms exist for computing a visibility graph from an input polygon~\cite{ghosh-mount-optimal-visibility}; however, comparatively little is known about the reverse direction: the so-called visibility graph \emph{recognition} and \emph{reconstruction} problems.

In this paper, we study \emph{vertex-vertex visibility graphs}, which are formed by visibility between pairs vertices of a polygon. Given a graph $G=(V,E)$, the visibility graph recognition problem asks if $G$ is the visibility graph of \emph{some} polygon. Similarly, the visibility graph reconstruction problem asks us to construct a polygon with $G$ as a visibility graph. Surprisingly, recognition of simple polygons is only known to be in PSPACE~\cite{everett-thesis}, and it is still unknown if simple polygons can be reconstructed in polynomial time.
Therefore, current solutions are typically for restricted classes of polygons. 
\ifFull
Even characterizations are only known for special classes of visibility graphs~\cite{esa-pseudo-polygons,orourke-streinu-1997}, and only four necessary conditions are known~\cite{ghosh-book} for standard visibility graphs.
\fi{}

\ifFull
\subsection{Special Classes}
\else
\subsection{Special Classes}
\fi 
A well-known result due to ElGindy~\cite{elgindy-thesis} is that every maximal outerplanar is a visibility graph and a polygon can be reconstructed from every such graph in polynomial time. Other special classes rely on a unique configuration of reflex and convex chains, which restrict visibility.
\ifFull
For instance, spiral polygons, which consist of a single reflex chain and a convex chain can be recognized in linear time~\cite{everett-spiral-1990}. Furthermore, tower polygons (also called funnel polygons), which consist of two reflex chains connected by a single edge, can be recognized in linear time~\cite{choi-funnel-1995}. Perhaps the most general class of visibility graphs that can be reconstructed in polynomial time are those from which from which a special class of $3$-matroids can be constructed in polynomial time~\cite{abello-kumar-2002}. Abello and Kumar~\cite{abello-kumar-1995} further showed that $2$-spirals fit into this class, and thus can be reconstructed in polynomial time.  Other special cases, such as convex polygons with a single convex hole~\cite{cai-everett-1995} can be reconstructed in polynomial time.

Related to our work here, Colley~\cite{colley-thesis-1991,colley-cccg-1992} showed that if each face of a maximal outerplanar graph is replaced by a clique on the same number of vertices, the resulting graph is a visibility graph of a uni-monotone polygon (monotone with respect to a single edge), and such a polygon can be reconstructed if the Hamiltonian cycle of the boundary edges is given as input. However, as noted in~\cite{evens-saeedi-2015}, not every uni-monotone polygon (even those with uniformly spaced vertices) has such a visibility graph. Evans and Saeedi~\cite{evens-saeedi-2015} further characterized terrain visibility graphs, which consist of a single monotone polygonal line.
\else
For instance, spiral polygons~\cite{everett-spiral-1990}, and tower polygons~\cite{choi-funnel-1995} (also called funnel polygons), can be reconstructed in linear time, and each consists of one and two reflex chains, respectively. $2$-spirals can also be reconstructed in polynomial time~\cite{abello-kumar-1995}, as can a more general class of visibility graphs related to $3$-matroids~\cite{abello-kumar-2002}.

For monotone polygons, Colley~\cite{colley-thesis-1991,colley-cccg-1992} showed that if each face of a maximal outerplanar graph is replaced by a clique on the same number of vertices, then the resulting graph is a visibility graph of some uni-monotone polygon (monotone with respect to a single edge), and such a polygon can be reconstructed if the Hamiltonian cycle of the boundary edges is known. However, not every uni-monotone polygon (even those with uniformly spaced vertices) has such a visibility graph~\cite{evens-saeedi-2015}. Finally, Evans and Saeedi~\cite{evens-saeedi-2015} characterized terrain visibility graphs, which consist of a single monotone polygonal line. 
\fi

\ifFull
Surprisingly little is known about visibility graphs of orthogonal polygons: we are only aware of results for \emph{orthogonal convex fans}, which consist of a single staircase and an extra vertex. These are also known as \emph{staircase polygons}. In particular, Abello and E\u{g}ecio\u{g}lu~\cite{abello-uniform-step-length} show that orthogonal convex fans with uniform side-length can be reconstructed in linear time; however, their construction relies on a simpler definition of visibility, which allows for blocking vertices. Furthermore, Abello~\etal~\cite{abello-convex-fans} show that general orthogonal convex fans are \emph{recognizable} in polynomial time, but reconstruction is still open even for this simple class of orthogonal polygons. Other algorithms for orthogonal polygons use different visibility representations, which seem to be easier to work with for orthogonal polygons.
\else
For orthogonal polygons, \emph{orthogonal convex fans} (also known as \emph{staircase polygons}), which consist of a single staircase and an extra vertex, can be recognized in polynomial time~\cite{abello-convex-fans}; however---strikingly---the \emph{only} class of orthogonal polygons known to be reconstructible in polynomial time is the staircase polygon with uniform step lengths, due to Abello and E\u{g}ecio\u{g}lu~\cite{abello-uniform-step-length}. %
\fi %
\ifFull %
 Other algorithms for orthogonal polygons use different visibility representations, which seem to be easier to work with for orthogonal polygons.
\subsection{Other Visibility Representations}
In the standard visibility graph of a polygon $P$, \ifFull also known as a vertex-vertex visibility graph,\fi{} edges are formed between vertices $u$ and $v$ if and only if the line segment $uv$ does not intersect the exterior of $P$. That is, the line of sight is allowed along edges, and through vertices of the polygon. As mentioned above, few results exist in the standard model of visibility. Instead, researchers have looked at other definitions of visibility, such as vertex-edge visibility, or edge-edge visibility. For orthogonal polygons in particular, edge-edge visibility graphs~\cite[\sec~7.3]{orourke-art-gallery-survey} consist of two disjoint trees and the polygon can be reconstructed when these trees ``mesh well''. Furthermore, orthogonal polygons can be reconstructed entirely from ``stabs''---for each vertex, we are given the edges hit by horizontal (or vertical) rays shot from the vertex~\cite{jackson-wismath-stabs} in $O(n\log n)$ time, if we are given the Hamiltonian cycle of the boundary edges with the input. See Asano~et al.~\cite{asano-survey} or Ghosh~\cite{ghosh-book} for a thorough review of results on visibility graphs.
\else
 Other algorithms for orthogonal polygons use different visibility representations such as vertex-edge or edge-edge visibility~\cite[\sec~7.3]{orourke-art-gallery-survey}, or ``stabs''~\cite{jackson-wismath-stabs}. See Asano~et al.~\cite{asano-survey} or Ghosh~\cite{ghosh-book} for a thorough review of results on visibility graphs.
\fi 

\ifFull
\subsection{Our Results}
\else
\subsection{Our Results}
\fi 
In this work, we investigate reconstructing polygons consisting of multiple uniform step length staircases.  
\ifFull To the best of our knowledge, the \emph{only} pure visibility reconstruction result for orthogonal polygons is for a single uniform step length staircase~\cite{abello-uniform-step-length}.\fi{} We first show that orthogonally convex polygons can be reconstructed in time $O(n^2m)$. We further show that reconstructing orthogonal uni-monotone polygons is fixed-parameter tractable, when parameterized on the number of the horizontal convex-convex boundary edges in the polygon. We also provide an $O(n^2m)$ time algorithm under reasonable alignment assumptions.
As a consequence of our reconstruction technique, we can also recognize the visibility graphs of these classes of polygons with the same running times. 

\section{Preliminaries}
Let $P$ be a polygon on $n$ vertices.
\ifFull We define visibility as follows.

\begin{definition}[visibility]
\label{definition:visibility}
Two points $p$ and $q$ are visible in polygon $P$ if line segment $pq$ does not intersect the exterior of $P$.
\end{definition}
\else
We say that a point $p$ \emph{sees} a point $q$ (or $p$ and $q$ are \emph{visible}) in polygon $P$ if the line segment $pq$ does not intersect the exterior of $P$.
\fi
Under this definition, visibility is allowed along edges and through vertices.

\ifFull
\begin{definition}[visibility graph]
A visibility graph $G_P=(V_P,E_P)$ of polygon $P$ has a vertex $v_p\in V_P$ for each vertex $p$ of $P$, and an edge $(v_p, v_q)\in E_P$ if and only if vertices $p$ and $q$ are visible in $P$.
\end{definition}
We further define a \emph{visibility graph} $G_P=(V_P,E_P)$ of polygon $P$ that has a vertex $v_p\in V_P$ for each vertex $p$ of $P$, and an edge $(v_p, v_q)\in E_P$ if and only if vertices $p$ and $q$ are visible in $P$. Due to the one-to-one relationship between polygon and graph vertices, we we reference the vertex of a visibility graph as we would a vertex of the polygon.
\else
\fi

\ifFull
We now give some basic definitions for graphs and visibility graphs in particular.
\else
\fi
For our visibility graph discussion, we adopt standard notation for graphs and polygons. In particular, for a graph $G=(V,E)$, we denote the \emph{neighborhood} of a vertex $v\in V$ by $N(v) = \{u\mid (v,u)\in E\}$, and denote the number of vertices and edges by $n=|V|$ and $m=|E|$, respectively. For a visibility graph $G_P=(V_P,E_P)$ of a polygon $P$, we call an edge in $G_P$ that is an edge of $P$ a \emph{boundary edge}. Other edges (diagonals in $P$) are {\em non-boundary edges}.

Finally, critical to our proofs is the fact that a maximal clique in $G_P$ corresponds to a maximal (in the number of vertices) convex region $R\subseteq P$ whose vertices are defined by vertices of $P$.  A vertex $v$ is called \emph{simplicial} if $N(v)$ forms a clique, or equivalently $v$ is in exactly one maximal clique. For our work here, we further adapt this definition for an edge. We say that an edge $(u,v)$ is {\em $1$-simplicial} if $N(u)\cap N(v)$ is a clique, or equivalently $(u,v)$ is in exactly one maximal clique\footnote{This is not to be confused with \emph{simplicial edges}, which are defined elsewhere\ifFull in the literature \fi{} to be edges $(u,v)$ such that for every $w\in N(u)$ and $x\in N(v)$, $w$ and $x$ are adjacent.}. The intuition behind why we consider $1$-simplicial edges is that, in orthogonal polygons with edges of uniform length, boundary edges between convex vertices are $1$-simplicial, with the vertices of the clique forming a rectangle. (See \fig~\ref{figure:maximal-clique}.)

\ifFull
For our running times, we rely on the following claim for $1$-simplicial edges.
\begin{cclaim}
\label{claim:simplicial-k-clique-time}
We can test if $(u,v)$ is $1$-simplicial and in a maximal $k$-clique in time $O(kn)$.
\end{cclaim}
\begin{proof} Compute $X = N(u)\cap N(v)$ in $O(n)$ time by marking vertices in $N(u)$ and checking for marked vertices in $N(v)$, then check that $|X| = k$ and $X$ is a clique in $O(kn)$ time.
\end{proof}
\else
Our running times depend on the following observation for $1$-simplicial edges.

\begin{observation}
\label{claim:simplicial-k-clique-time}
We can test if $(u,v)$ is $1$-simplicial and in a maximal $k$-clique in time $O(kn)$.
\end{observation}
\fi{}

\begin{figure}[!tb]
\begin{center}
\subfloat[]{\includegraphics[scale=0.60]{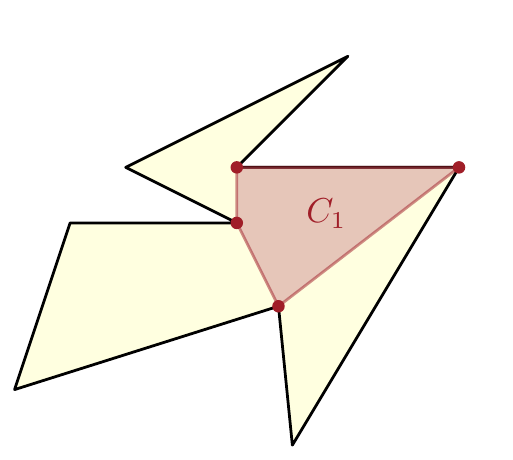}}
\subfloat[]{\includegraphics[scale=0.60]{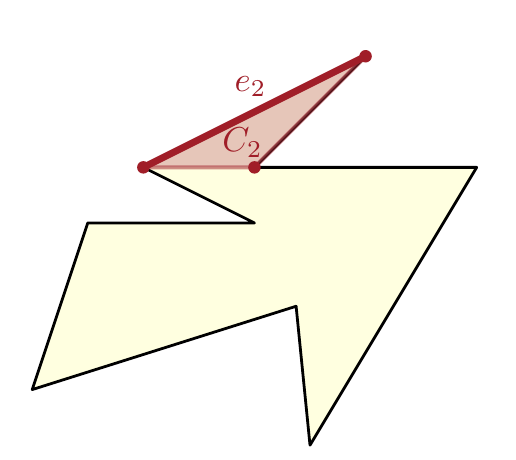}}
\subfloat[]{\includegraphics[scale=0.60]{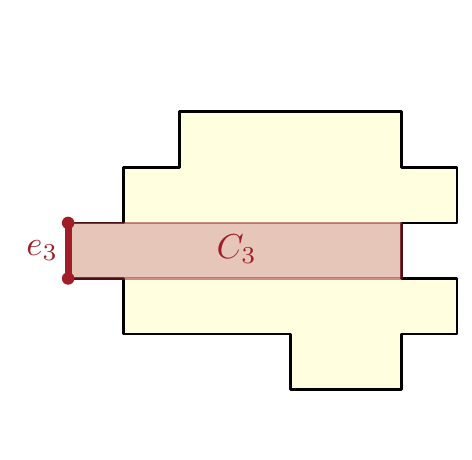}}
\caption{Maximal convex regions on vertices of polygons are maximal cliques in visibility graphs. (b)-(c) $1$-simplicial edges are in exactly one maximal clique.}
\label{figure:maximal-clique}
\end{center}
\end{figure}

\ifFull

\subsection{Uniform Step Length Staircase Polygons}
We briefly mention the construction of uniform step length staircase polygons, which were shown to be reconstructible in polynomial time by Abello and E\u{g}ecio\u{g}lu~\cite{abello-uniform-step-length}.

Let $P$ be a staircase polygon on $n$ vertices $\{v_0,v_1,\ldots,v_{n-1}\}$, where $v_0v_1\cdots v_{n-1}$ is the clockwise ordering of the vertices around the polygon. Notice that there is one unique vertex (let's say $v_0$) that sees all vertices, and has two convex boundary neighbors $v_1$ and $v_{n-1}$. Then, there are $(n-3)$ vertices, where each convex vertex $v_i$ only sees $v_0$ and its two boundary neighbors $v_{i-1}$, and $v_{i+1}$, and each reflex vertex $v_j$ sees $v_0$, convex neighbors on the staircase $v_{j-1}$ and , $v_{j+1}$, and reflex vertices $v_{j-2}$ and $v_{j+2}$, which block $v_j$'s view to other reflex vertices by Abello and E\u{g}ecio\u{g}lu's definition~\cite{abello-uniform-step-length}.

For the construction, we first assume that the polygon has more than $12$ vertices to avoid special cases. We first identify $v_0$, which is determined completely by its degree $(n-1)$. We then reconstruct by separating the vertices into convex and reflex vertices and ordering them along the staircase. Each convex vertex has exactly 3 neighbors, each reflex has exactly 5 neighbors. Then the order of the vertices is determined completely by the neighborhoods of each convex vertex. In particular, two convex vertices $v_i$ and $v_j$ are consecutive along the boundary if and only if $|N(v_i)\cap N(v_j)\setminus\{v_0\}| = 1$. If the vertices are neighbors, then they are $v_1,v_2$ and $v_{n-2},v_{n-1}$. Otherwise, $N(v_i)\cap N(v_j)\setminus\{v_0\}$ contains only the unique reflex vertex between $v_i$ and $v_j$ along the staircase. This can be done in $O(n+m)$ time.

Note that by changing the visibility model, things become only slightly more complex. All reflex vertices see all other reflex vertices, and vertices $v_1$ and $v_{n-1}$ see all reflex vertices as well. However, it is still true that each convex vertex on the staircase has only three neighbors. Therefore, reconstruction can still be done by ordering the convex vertices along the staircase in $O(n+m)$ time.
\fi




\ifFull
\darren{TODO: mention counting argument of Abello and E.}
\fi 

\ifNewSchool
\section{Attempt at unifying double staircases and monotone staircases}

We begin with the following lemma that allows us to identify every tab edge. 
\begin{lemma}
An edge $(u,v)$ in a monotone staircase polygon is a tab edge iff it is 1-simplicial and is in a maximal 4-clique.
\end{lemma}

To reconstruct monotone staircase polygon, we iteratively remove tab cliques until the remaining graph is $K_4$. (While the above lemma is sufficient to identify tab edges in the input monotone staircase polygon, it does not hold for truncated tabs, i.e., for detecting newly created tabs once original tab cliques are removed. We present how to iteratively remove tab cliques later...) 

\newcommand{\pyramid}{pyramid\xspace}
\newcommand{\pyramids}{pyramids\xspace}

\newcommand{\topV}{top vertex\xspace}
\newcommand{\topVs}{top vertices\xspace}
\newcommand{\bottomV}{bottom vertex\xspace}
\newcommand{\bottomVs}{bottom vertices\xspace}
\newcommand{\fixed}{fixed\xspace} 
\newcommand{\companion}{companion\xspace} 

This decomposes the monotone polygon into axis-aligned rectangles whose contact graph is a tree $T$ (see \fig~\ref{figure:monotone-staircase-and-contact-tree} for an example). We root $T$ at the rectangle containing the base edge (the base rectangle), consequently,  each leaf is a rectangle containing a tab (a tab rectangle). The decomposition process associates rectangles (quadruplets of vertices) with each node of $T$. Observe that each such rectangle consists of two {\em \topVs} that are convex and two {\em \bottomVs} that are either both reflex or are both convex {\em base} vertices. 
Given a \topV (resp., \bottomV) $p$, we say the other \topV (resp., \bottomV) in the rectangle, denoted by $\bar{p}$, is its \companion vertex.

Lemma~\ref{lemma:convex-reflex} describes how we can distinguish \topVs from \bottomVs, and identify tab, base and dent vertices efficiently. Observe that since the monotone staircase polygon contains unit-length edges, the tree decomposition and the differentiation of vertices into \topVs and \bottomVs defines the $y$-coordinate of each vertex in the polygon. Thus it remains to determine the $x$-coordinate of each polygon vertex. 

We say the contact tree $T$ is {\em \fixed} if at each node of $T$, the left to right order of the children is determined.
Given a \fixed contact tree $T$, Lemmas~\ref{} and \ref{lemma:x-coords} show that the $x$-coordinate of each vertex is uniquely defined by whether the vertex is on the left or the right side of the rectangle in the rectangle decomposition. 
Therefore, to compute the $x$-coordinates, we need to color all vertices with two colors ($0$ and $1$), denoting if it belongs to the left or the right side the corresponding rectangle. 

Let us describe a high level picture of how we perform the coloring.

\newcommand{\dependent}{orientation dependent\xspace}
\newcommand{\dependency}{orientation dependency\xspace}

We say a vertex $p$ and $q$ are {\em \dependent}, 
if they must be colored in two different colors. For example, every vertex $p$ and its companion $\bar{p}$ are \dependent. In Lemma~\ref{lemma:dependency} we identify all \dependent pairs of vertices.

In Lemma~\ref{} we show how given a \fixed contact tree $T$, we identify the colors of all vertices whose colors are a function of the visibility, imposed by the structure of the tree. 

We use the above information as follows. First, we construct a \dependency graph $G'=(V, E')$, where edge $(u,v) \in E'$ if $u$ and $v$ are \dependent. Next, for each possible \fixed version of $T$ (i.e. for each ordering of leaves in $T$), we color a subset of vertices of $G'$ using the structure of the \fixed tree. Finally, we try to 2-color $G'$, using these initial colors on a subset of nodes. If there exists at least one ordering of leaves that allows 2-coloring of $G'$, we use the coloring and Lemma~\ref{lemma:x-coords} to determine the $x$-coordinates of every vertex.

\subsection{Details}

\newcommand{\visq}[1]{\ensuremath{vis_{q, \bar{q}}(#1)}}

\begin{lemma}\label{lemma:dependency}
  Consider any two pairs of vertices $\{p, \bar{p}\}$ and  $\{q, \bar{q}\}$.  Let $\visq{p}$ be a set of vertices among ${q, \bar{q}}$ that $p$ sees. Then the following cases define all possible \dependent relationships between $\{p, \bar{p}\}$ and $\{q, \bar{q}\}$:

\begin{enumerate}
    \item $\visq{p} = \visq{\bar{p}}$: $p$ is not \dependent on either $q$ or $\bar{q}$; $\bar{p}$ is not \dependent on either $q$ or $\bar{q}$.
    \item $\visq{p} \neq \visq{\bar{p}}$:
      \begin{enumerate}
	\item \label{item:vis11} $|\visq{p}| = |\visq{\bar{p}}| = 1$:  $p$ and $\visq{p}$ are \dependent and so are $\bar{p}$ and $\visq{\bar{p}}$.

	\item \label{item:vis02} $|\visq{p}| = 0$, $|\visq{\bar{p}}| = 2$: impossible.
	\item \label{item:vis20} $|\visq{p}| = 2$, $|\visq{\bar{p}}| = 0$: impossible.
	\item \label{item:vis01} $|\visq{p}| = 0$, $|\visq{\bar{p}}| = 1$: $\bar{p}$ and $\visq{\bar{p}}$ are \dependent.
	\item \label{item:vis10} $|\visq{p}| = 1$, $|\visq{\bar{p}}| = 0$: $p$ and $\visq{p}$ are \dependent.

	\item \label{item:vis12} $|\visq{p}| = 1$, $|\visq{\bar{p}}| = 2$: $p$ and $\visq{p}$ are \dependent.
	\item \label{item:vis21} $|\visq{p}| = 2$, $|\visq{\bar{p}}| = 1$: $\bar{p}$ and $\visq{\bar{p}}$ are \dependent.
      \end{enumerate}
\end{enumerate}
\end{lemma}

\begin{figure}
  \includegraphics[scale=0.60]{figures/vis11}
\quad 
  \includegraphics[scale=0.60]{figures/vis02} \\
  \includegraphics[scale=0.60]{figures/vis01}
\quad 
  \includegraphics[scale=.60]{figures/vis12}
\quad 
\caption: Cases~\ref{item:vis11}, \ref{item:vis02}, \ref{item:vis01} and \ref{item:vis12}. \\
\\
\label{figure:dependencies}
\end{figure}

\begin{proof}
Case 1: If $\visq{p} = \visq{\bar{p}}$ then both $p$ and $\bar{p}$ either (a) see both $q$ and $\bar{q}$ or (b) see neither of them. Since both cases are vertically (?) symmetric, coloring of $p$ or $\bar{p}$ has no effect on coloring of $q$ or $\bar{q}$.

Let's consider Case 2: $\visq{p} \neq \visq{\bar{p}}$. Observe, that if $y(p) = y(q)$, then in monotone staircases, this falls under Case 1. Thus, for the rest of the proof, w.l.o.g. we assume that $y(p) > y(q)$.

Let $p' = (x(p), y(q))$ and $\bar{p}' = (x(\bar{p}), y(q))$ be the projections of $p$ and $\bar{p}$  on the line segment $q\bar{q}$. We know that the rectangle $p\bar{p}\bar{p}'p'$ is empty (refer to Fig.~\ref{figure:dependencies} for illustration).

Case~\ref{item:vis11}: Let $q = \visq{p}$ and $\bar{q} = \visq{\bar{p}}$. Assume $p$ and $q$ are both assigned color $b$, then $\bar{p}$ and $\bar{q}$ must be assigned color $1-b$. Then the quadrilateral $qp\bar{p}\bar{q}$ must be empty, implying that $q$ is visible from $\bar{p}$ and $\bar{q}$ is visible from $p$ and contradicting that $|\visq{p}| = |\visq{\bar{p}}| = 1$

Case~\ref{item:vis02}: There are two cases to consider.  First case: $q$ and $p$ are assigned the same color. Since $\bar{p}$ sees both $q$ and $\bar{q}$ we know that the $\triangle q\bar{p}\bar{q}$ is empty and, consequently, $\triangle \bar{p}'p\bar{q} \subseteq \triangle q\bar{p}\bar{q}$ is also empty. Then the $p'p\bar{p}\bar{q} = \square p'p\bar{p}\bar{p}' \cup \triangle \bar{p}'\bar{p}\bar{q}$ is also empty, implying that $\bar{q}$ is visible from $p$ -- a contradiction.  The second case ($\bar{q}$ and $p$ are assigned the same color) is proven to be impossible symmetrically.

Case~\ref{item:vis20}: Symmetric to Case~\ref{item:vis02}.

Case~\ref{item:vis01}: Let $\bar{q} = \visq{\bar{p}}$ and assume $\bar{q}$ and $\bar{p}$ are assigned the same color. Then $\triangle \bar{p}'\bar{p}\bar{q}$ is empty.  Then the polygon $p'p\bar{p}\bar{q} = \square p'p\bar{p}\bar{p}' \cup \triangle \bar{p}'\bar{p}\bar{q}$ is also empty, implying that $\bar{q}$ is visible from $p$ -- a contradiction that neither $q$ nor $q'$ are visible from $p$.

Case~\ref{item:vis10}: Symmetric to Case~\ref{item:vis01}.

Case~\ref{item:vis12}: Let $q = \visq{p}$ and assume $p$ and $q$ are assigned the same color. Then the quadrilateral  $qp\bar{p}\bar{q}$ must be empty, implying that $\bar{q}$ is visible from $p$  and contradicting that $|\visq{p}|  = 1$.

Case~\ref{item:vis21}: Symmetric to Case~\ref{item:vis12}.

\end{proof}

\newcommand{\size}[0]{\mathrm{size}}

\begin{lemma}\label{lemma:x-coords}
Let $v$ be a node in $T$. Let $\size(v)$ be the size of the subtree rooted at $v$ (including $v$). Let $p$ and $q$ be the two \topVs in the rectangle associated with $v$. Then $|x(p) - x(q)| = 2\cdot \size(v) - 1$.
\end{lemma}

\begin{proof}
By induction. {\bf Base case:} $v$ is a leaf. $|x(p)-x(q)| = 1 = 2\cdot \size(v) - 1$.

{\bf Inductive case:} Let $u_1, ..., u_k$ be the children of $v$ and let  $p_i$ and $q_i$ be the \topVs of rectangle associated with $u_i$. Then 
\begin{eqnarray*}
|x(p) - x(q)| &=& 1 + \sum_{i=1}^k \left(1+|x(p_i)-x(q_i)|\right)       \\
	      &=& 1 + \sum_{i=1}^k \left(1+2\cdot \size(u_i)-1\right) \\
	      &=& 1 + 2\sum_{i=1}^k \size(u_i) \\
	      &=& 1 + 2\cdot (\size(v)-1) \\
	      &=& 2\cdot \size(v) - 1. 
\end{eqnarray*} 
\end{proof}

\begin{lemma}\label{lemma:fixed-color}
Let $v$ be a node in the contact tree $T$ with 2 or more children and let $u$ and $w$ be the leftmost and rightmost leaves of the subtree of $T$ rooted at $v$. Let $p$  be a bottom vertex of $v$. Then $p$ must be colored $0$ iff it sees some \topV of $u$. Similarly, $p$ must be colored $1$ iff it sees some \topV of $w$.
\end{lemma}

\begin{proof}
To be done, but easy. 
\end{proof}

\else
\newcommand{\pyramid}{pyramid\xspace}
\newcommand{\pyramids}{pyramids\xspace}

\newcommand{\topV}{top vertex\xspace}
\newcommand{\topVs}{top vertices\xspace}
\newcommand{\bottomV}{bottom vertex\xspace}
\newcommand{\bottomVs}{bottom vertices\xspace}
\newcommand{\fixed}{fixed\xspace} 
\newcommand{\companion}{companion\xspace} 

\fi 

\newif\ifOLD
\OLDfalse
\ifOLD

We say two rectangles $R_1$ and $R_2$ in the decomposition are \emph{orientation-fixed} if a \topV from sees a \bottomV of another.  Then these two rectangles must be oriented so that the visibility constraint holds. That is, fixing an orientation of one rectangle fixes the orientation of the other.

Note that, every rectangle is orientation-fixed with some leaf rectangle (as there is a path from $v\in T$ to a leaf that only goes through descendants). Therefore, it is sufficient to fix the ordering (and orientation) of the leaves, which induces an ordering/orientation of the tree. Note that for double staircases, $T$ is a path, and the root rectangle is orientation-fixed with every other rectangle (a base vertex is seen by every \topV of the other rectangles). Hence, orienting the only tab rectangle determines the positions of the convex vertices on the double staircase.

 In \sec~\ref{sec:xxx} we show how to determine the ordering of leaves.

\begin{corollary}
The two \topVs of each rectangle are orientation-fixed. 
\end{corollary}

Given a rectangle $R$ and a top (resp., bottom) vertex $t$ (resp., $b$), we denote the other top (resp., bottom) vertex $\bar{t}$ (resp., $\bar{b}$).

\begin{definition}
  Given a vertex $p$ and $q$, we say $p$ and $q$ are {\em orientation-fixed} if $p$ sees $q$, but does not see $\bar{q}$. Since $p$ does not see itself, the above definition implies that $p$ and $\bar{p}$ are orientation-fixed.
\end{definition}

\begin{lemma}
  Let $u$ and $v$ be two nodes in the tree decomposition $T$ and let $R_u$ and $R_v$ be the rectangles corresponding to $u$ and $v$. If $p \in R_u$ and $q \in R_v$  are orientation fixed, then either $u$ or $v$ are the lowest common ancestor of $u$ and $v$. 
\end{lemma}

\begin{proof}
  Assume $u$ and $v$ are not ancestors of each other. If $u$ and $v$ are siblings, then either $u$ sees both $v$ and $\bar{v}$ or does not see either of them. Similarly, either $v$ sees both $u$ and $\bar{u}$ or does not see either of them. If they are in different subtrees and are not siblings then $u$ sees neither $v$ nor $\bar{v}$ and $v$ sees neither $u$ nor $\bar{u}$.
\end{proof}

\begin{definition}
We say two rectangles $R_1$ and $R_2$ are orientation-fixed if at least one vertex of $R_1$ is orientation-fixed with at least one vertex of $R_2$.
\end{definition}

\begin{lemma}
Given a node $v \in T$ with $k$ children $u_1, \dots, u_k$, rectangle $R_v$ is orientation fixed with $R_{u_1}$ and $R_{u_k}$.
\end{lemma}
\begin{proof}
Let $b_v^l \in R_v$ be the left \bottomV and let $t_{u_1}^l$ and $t_{u_1}^r$ be the left and right \topVs of $u_1$. Then it is easy to see that $b_v^l$ sees $t_{u_1}^r$, but does not see $t_{u_1}^l$.

Similarly, $b_v^r \in R_v$ sees $t_{u_k}^l$, but does not see $t_{u_k}^r$.  

\end{proof}

\begin{lemma}
All rectangles on the path from root to the right-most leaf and to the left-most leaf are orientation-fixed. 
\end{lemma}

\begin{proof}
By induction on level of a node.
{\bf Base case:} parent rectangle of the lowest-level leaf is orientation-fixed with the right-most and left-most leaf node by above lemma.

{\bf Inductive step:} Assume the the rectangle at every node $v$ at level $l$ are orientation-fixed with all the nodes on the paths from $v$ to the right-most and left-most leaves. Consider a node $u$ at level $l+1$. By the above lemma, $u$ is orientation-fixed with $v_l$, the left-most child of $v$, and $v_r$, the right-most child of $v$, which lie on the path to the right-most and left-most leaves. 
\end{proof}

\begin{lemma}
  Let $u$ and $v$ be two nodes in the tree decomposition $T$, with specific order among children at each node, and let $R_u$ and $R_v$ be the rectangles corresponding to $u$ and $v$. If $p \in R_u$ and $q \in R_v$  are orientation fixed, then $|x(p)-x(q)| = f(d(u,v)$, where $d(u,v)$ is the distance between $u$ and $v$ in the properly defined Euler tour on $T$.
\end{lemma}

We construct graph $G = (V,E)$, where $V$ is a set of the vertices of the polygon and $(u, v) \in E$ if vertices $u$ and $v$ are orientation fixed in the polygon. 

\begin{lemma}
If at least one vertex $v$ of a subtree $T'$ is orientation fixed with a vertex in a different subtree $T''$, and $v$ is not a \topV of the dent rectangle of $T'$ and is not a \bottomV of the tab rectangle, then all vertices in $T'$ are orientation-fixed to $T''$.
\end{lemma}

\begin{lemma}
$G$ is bipartite.
\end{lemma}

\begin{proof}
Let $(p, q) \in E$ and let $p \in R_p$ and $q \in R_q$. Then either ($p$ is the left vertex of $R_p$ and $q$ is a right vertex of $R_q$) or ($p$ is the right vertex of $R_p$ and $q$ is the left vertex of $R_q$). It might be hard to prove this and might be unnecessary...
\end{proof}

 There are $O(k!2^k)$ such orderings (and orientations) for all leaf rectangles, where $k$ is the number of tabs.

In fact, we can extend this idea to paths in the monotone staircase:

\begin{lemma}
The base rectangle of a monotone staircase is orientation-fixed with all rectangles on the paths in $T$ to the leaves containing the left-most and right-most tabs
in the monotone staircase.
\end{lemma}
\begin{proof}
The left base vertex sees all (right) convex vertices in the left-most tab's double staircase, and the right base vertex sees all (left) convex
vertices in the right-most tab's double staircases.
\end{proof}

\begin{lemma}
\label{lemma:convex-reflex}
Given a decomposition of the monotone polygon into rectangles and the corresponding contact graph $T$ with associated quadruplets of vertices with each node,  we can distinguish \topVs from \bottomVs in $O(nm)$ (?) time. Moreover, we can identify tab vertices, base vertices and dent vertices.  
\end{lemma}
\begin{proof}
To be written... Relies on Lemma~\ref{lemma:topV} to identify \topVs. Can identify tab vertices via 1-simplicial 4-cliques. Can identify base vertices by Lemma~\ref{lemma:baseV}.
\end{proof}

\begin{lemma}
\label{lemma:topV}
With the exception of the tab rectangles and the base rectangle, each \topV sees exactly one base vertex or exactly one tab vertex.
\end{lemma} 
\begin{proof}
To be written...
\end{proof}

\begin{lemma}
\label{lemma:baseV}
In the base rectangle, base vertices see at least one tab vertex, while \topVs see no tab vertex.
\end{lemma}
\begin{proof}
To be written...
\end{proof}

We note that it is possible to reconstruct the monotone staircase by defining left-to-right ordering the children of each node in $T$ and by specifying a left-right orientation of the top edge of the node's corresponding rectangle (a tab that was removed at some point during the algorithm). That is, it is possible to place the rectangles in the correct position and orientation so that the visibility graph of the output monotone staircase is correct.

Since each edge of the polygon is unit length, the information provided in the above process is sufficient to determine the $y$-coordinate of each polygon vertex. Thus, it remains to determine the $x$-coordinates. 

Define the {\em \topVs} of a rectangle and the {\em \bottomVs}. Note, every base vertex and a dent vertex is a \bottomV, every tab vertex is a \topV. Equivalently, every \topV is convex, while \bottomV can be either reflex, or in case of the base vertex, it's convex.

Let {\em \pyramid} be a double staircase that starts at a base or a dent vertex, goes up to a tab edge, and then down to a base or dent vertex. Observe, that there is a unique \pyramid associated with each of $k$ tabs and in case of a double staircase, the only \pyramid consists of two staircases of equal lengths, while in case of monotone staircase, each of $k$ \pyramids consists of two staircases of not necessarily equal lengths. 

Observe that the decomposition of the polygon into rectangles and construction of $T$ associates quadruples of vertices (which define rectangles of the \pyramids) with each node of $T$. The structure of the tree defines the structure of the \pyramids. 

Therefore, to determine the $x$-coordinates of all the polygon vertices, it is sufficient to define a (valid/correct/satisfying/appropriate?) left-to-right ordering of \pyramids and a (valid/correct/satisfying/appropriate?) assignment of the vertices  to the left or right staircase of the \pyramid that they belong to.  

Note that due to the one-to-one correspondence between the leaves of $T$ and the \pyramids, it is sufficient to focus on the left-to-right ordering of leaves in $T$. (This step is unnecessary for a double staircase polygon, because there is only one tab/leaf.)

Finally, we describe how to associate vertices of the polygon with the left/right staircases of each \pyramid (equivalently, define valid (?) left-right orientation of the vertices within each rectangle)

\fi

\section{Uniform-Length Orthogonally Convex Polygons}
We first turn our attention to a restricted class of orthogonal polygons that have only uniform-length (or equivalently, unit-length) edges.
Let $P$ be an orthogonal polygon with uniform-length edges such that no three consecutive vertices on $P$'s boundary are collinear, and further let $P$ be \emph{orthogonally convex}\footnote{That is, any two points in $P$ can be connected by a staircase contained in $P$.}. We call $P$ a \emph{uniform-length orthogonally convex polygon} (UP).
Note that every vertex $v_i$ on $P$'s boundary is either convex or reflex.
We call boundary edges between two convex vertices in a uniform-length orthogonal polygon $P$ \emph{tabs} and a tab's endvertices \emph{tab vertices}.
%
We reconstruct the polygon by computing the clockwise ordering of vertices of the UP.

Note that the boundary of a UP consists of four tabs connected via staircases. For ease of exposition, we imagine the UP embedded in $\mathbb{R}^2$ with polygon edges axis-aligned. We call the tab with the largest $y$-coordinate the north tab, and we similarly name the others the south, east, and west tabs. We similarly refer to the four boundary staircases as northwest, northeast, southeast, and southwest.

We only consider polygons with more than $12$ vertices, which eliminates many special cases.\ifFull\footnote{While simple rules may extend to these smaller cases, recognition and reconstruction can be done in constant time by generating all possible uniform-length orthogonally convex polygons, computing the visibility graph for each one, and checking if the visibility graph is isomorphic to our input graph.} \else{} Smaller polygons can be solved in constant time via brute force.\fi

We first introduce several structural lemmas which help us identify convex vertices in a UP, which is key to our reconstruction.

\ifFull
\begin{observation}
In a visibility graph of a UP, there is exactly one maximal clique containing all reflex vertices. Moreover, this clique contains no tab vertex. 
\end{observation}
\fi

\begin{lemma}
\label{lemma:convex-neighbor-one-clique}
For every convex vertex $u$ in a UP there is a convex vertex $v$, such that $(u,v) \in E_P$ and $(u,v)$ is $1$-simplicial.
\end{lemma}
\begin{proof}
If $u$ is a tab vertex, then the other tab vertex $v$ is also convex and $(u,v)$ is $1$-simplicial.
Otherwise, without loss of generality, suppose that $u$ is on the northwest staircase. Then there is a convex vertex $v$ on the southeast staircase that is visible from $u$. Edge $(u,v)$ is in exactly one maximal clique, consisting of $u$, $v$, the reflex vertices within the rectangle $R$ defined by $u$ and $v$ as the opposite corners, and any other corners of $R$ that are convex vertices of the polygon.
\end{proof}

\begin{lemma}
\label{lemma:convex-neighbor-two-cliques}
In a UP, if $u$ or $v$ is a reflex vertex, then edge $(u,v)$ is not $1$-simplicial. 
\end{lemma}

\begin{proof}[\sketch\footnote{Full proofs may be found in Appendices~\ref{section:unit-orthogonal-appendix} and~\ref{section:monotone-staircase-appendix}.}]
If both $u$ and $v$ are reflex, then $(u,v)$ is in one maximal clique consisting of only reflex vertices and another one that includes some convex vertex $w$. If one of $u$ or $v$ is convex, there exist two convex vertices $w$ and $w'$, forming two distinct maximal cliques with $(u,v)$. \ifAppendix See Appendix~\ref{subsection:two-cliques} for the full proof.\fi{}
\end{proof}

\ifFull
Thus, by testing if each edge is in more than one maximal clique, we can determine which vertices are convex and reflex.
\begin{lemma}
We can identify all convex and reflex vertices in a visibility graph of a UP in $O(n^2m)$ time.
\end{lemma}
\begin{proof}
For each edge, compute if it is $1$-simplicial $O(n^2)$ time. If so, its endvertices are convex. Any vertices not assigned to be convex are reflex.
\end{proof}
\else
Lemma~\ref{lemma:convex-neighbor-two-cliques} states that only edges between convex vertices can be $1$-simplicial. Hence it allows us to identify all convex vertices, by checking for each edge $(u,v)$ if $N(u)\cap N(v)$ is a clique in $O(n^2)$ time, leading to the following lemma.
\begin{lemma}
We can identify all convex and reflex vertices in a visibility graph of a UP in $O(n^2m)$ time.
\end{lemma}
\fi

\ifFull
We now divide the class of uniform-length orthogonal polygons into two classes of polygons, which we consider in turn. We call these classes \emph{regular} and \emph{irregular}.
\fi{}

\ifFull
\begin{definition}[regularity]
We call a UP \emph{regular} if each of its staircase boundaries have the same number of vertices. Otherwise, we call it \emph{irregular}, consisting of two \emph{long} and two \emph{short} staircases.
\end{definition}
\else
We say a UP is \emph{regular} if each of its staircase boundaries have the same number of vertices. Otherwise, we call it \emph{irregular}, consisting of two \emph{long} and two \emph{short} staircases.
\fi{}
\ifFull
We first consider regular uniform-length orthogonal polygons, which are simpler to recognize and reconstruct than their irregular counterparts.

\subsubsection{Regular Uniform-Length Orthogonal Polygons}

\begin{lemma}
\label{lemma:regular-cliques-convex-vertices}
In a regular uniform-length orthogonally-convex polygon, there are only two maximal 8-cliques that each contain four convex vertices.
\end{lemma}
\begin{proof}
No two convex vertices on the same staircase are visible, therefore, each convex vertex in a clique must come from a different staircase. Consider the two maximal cliques containing the tabs. These are 8-cliques and contain four convex vertices. Any other clique containing four convex vertices contain eight reflex vertices, and therefore have twelve vertices total.
\end{proof}

Let us call the two $8$-cliques containing the tabs \emph{primary} cliques.

\begin{lemma}
\label{lemma:regular-8-edges}
In a regular uniform-length orthogonally-convex polygon, we can find a set of eight edges that contains the four tabs in polynomial time.
\end{lemma}
\begin{proof}
In the proof of Lemma~\ref{lemma:regular-cliques-convex-vertices} we showed that there are two maximal 8-cliques containing the tabs (the primary cliques). If we take all diagonal edges between convex vertices in these cliques, there are eight such edges.
\end{proof}

\ifFull
\begin{figure}[!tbh]
\begin{center}
\includegraphics[width=7cm]{figures/regular-convex-convex-edges-smaller}
\caption{After computing the eight convex-convex edges containing the four tabs. There are four remaining convex-convex edges between these vertices.}
\label{figure:regular-convex-convex-edges}
\end{center}
\end{figure}
\else
\begin{figure}[!tbh]
\begin{center}
\includegraphics[width=\linewidth]{figures/regular-convex-convex-edges-smaller}
\caption{After computing the eight convex-convex edges containing the four tabs. There are four remaining convex-convex edges between these vertices.}
\label{figure:regular-convex-convex-edges}
\end{center}
\end{figure}
\fi{}

\begin{lemma}
In a regular unit-orthogonal polygon, we can find the four tabs in polynomial time.
\end{lemma}
\begin{proof}
First find the eight candidates as in Lemma~\ref{lemma:regular-8-edges}. Then we determine which edges are the 2-convex boundary edges as follows.

There are four remaining edges $R$ between all pairs of these convex vertices. See \fig~\ref{figure:regular-convex-convex-edges}. Consider any two of these edges. If two vertices $u$ and $v$ can see each other by a vertical or horizontal line of sight, then the edges cross in the polygon. Furthermore, there is a unique path of three vertices from $u$ to $v$ beginning and ending with an edge in $R$ the middle edge on this path is a tab.

If the edge $(u,v)$ is in two or more maximal cliques, then it is a vertical edge. Hence we can iterate over all pairs of edges in $R$, and find the four that overlap, and find all four tabs.
\end{proof}

We place the tabs as follows: We pick one primary clique, and choose its tabs to be the bottommost and topmost edges of the polygon. Since we know the vertical edges connecting vertices on different tabs, we can orient these edges in a clockwise direction. The leftmost and rightmost edges are fully specified by their visibility with the already placed edges: the tail of the previous edge sees the head of the next.

In the remainder of this section, we show to place the remaining vertices in order along the staircases connecting the tabs. Let us call each maximal clique containing four convex vertices an \emph{elementary} clique. We first show how to compute the remaining elementary cliques and compute its order along the staircase. Then we show how assign each convex vertex to a staircase, . With this ordering, we show how to place the reflex vertices as well.

We use the following properties of the elementary cliques:

\begin{observation}
In a regular unit-orthogonal polygon, each convex vertex is in exactly one elementary clique.
\end{observation}

\begin{lemma}
\label{lemma:elementary-cliques}
Each primary clique shares a 4-clique with exactly one other elementary clique and each non-primary elementary clique shares a 4-clique with exactly two other elementary cliques.
\end{lemma}
\begin{proof}
A primary clique has four convex vertices and four reflex vertices. Each one of the reflex vertices have another convex neighbor along the boundary. Together these four convex vertices form a maximal clique, and hence are an elementary clique.

We sketch a proof for the second claim by induction.

Base case: Let $C_0$ be a primary clique and let $C_1$ be the elementary clique that shares a 4-clique with $C_0$. Elementary clique $C_1$ contains four convex vertices and eight reflex vertices. Four of the reflex vertices are shared with $C_0$, and the remaining reflex vertices each have another convex neighbor along the boundary which form another elementary clique $C_2$.
\end{proof}

Furthermore, Lemma~\ref{lemma:elementary-cliques} induces an ordering of the cliques along the staircases. Refer to these cliques in order $C_0,\ldots,C_{k-1}$, where $C_0$ and $C_{k-1}$ are primary cliques.

We can therefore iteratively peel away the elementary cliques, beginning at primary clique $C_0$. We first peel off all convex vertices of the clique, leaving a 4-clique between remaining reflex vertices. This 4-clique is then shared with exactly one elementary clique $C_1$. We can then compute the remaining four convex vertices that share a clique with the 4-clique, and remove them, the reflex vertices from $C_0$, and a new 4-clique between reflex vertices remains. We continue this peeling until we reach clique $C_{k-1}$.

We now have the elementary cliques, and in particular, the four convex vertices from each elementary clique. We further have an order in which the vertices appear on each staircase. What remains is to assign each convex vertex to a staircase.

\begin{lemma}
\label{lemma:convex-staircase-order}
We can in polynomial time, assign each convex vertex to its staircase.
\end{lemma}
\begin{proof}
Let $V_i$ be the four convex vertices from elementary clique $C_i$. Let $v$ be a convex vertex in $V_i$. Then $v$ has a visibility edge to exactly one tail of some tab. Call this boundary edge $e_i$, then $v$ is on the staircase $S_{i+1,i+2}$.
\end{proof}

And further, we know the order of these vertices along each staircase.

\begin{figure}[!tbh]
\begin{center}
\includegraphics[scale=1.0]{figures/exchangable-vertices}
\caption{Left: The unique maximal clique formed determines the reflex vertices on one side, plus 2 vertices on adjacent staircases. Right: The vertices not on the staircase have the same neighborhood as the upper- and lowermost reflex vertices on the staircase and therefore can be freely exchanged.}
\label{figure:exchangable-vertices}
\end{center}
\end{figure}

\begin{lemma}
\label{lemma:reflex-staircase-order}
We can in polynomial time, assign each reflex vertex to its staircase in order.
\end{lemma}
\begin{proof}
For each tab, consider the edge from it's tail vertex to the head of the next tab in clockwise order. This edge is contained in exactly one maximal clique, containing all reflex vertices on the staircase between the convex vertices, plus two additional reflex vertices--one from each adjacent staircase.

We can compute the reflex vertices on the staircase as follows. For each convex vertex on the staircase, we compute which reflex vertices are its neighbors. This uniquely determines the positions of all but four reflex vertices (there are two choices for the first/last reflex vertices on the staircase). However, the first (resp., last) vertex has the same neighborhood as the reflex vertex on the previous (resp., next) staircase. Therefore these vertices may be freely chosen for either position. Choose one for the first/last position. We have now constructed a staircase. perform the operation for the remaining three staircases. 
\end{proof}

\begin{theorem}
\label{theorem:regular-reconstruction}
In polynomial time, we can reconstruct a regular uniform-length orthogonal polygon $P$ from its visibility graph $G_P$.
\end{theorem}
\begin{proof}
Following the procedures from this section, we compute all the boundary edges in clockwise order. Then by Observation~\ref{observation:reconstruction} we can uniquely place each vertex.
\end{proof}
\else
We restrict our attention to irregular uniform-length orthogonally convex polygons (IUPs); however, similar methods work for their regular counterparts.
\fi

\subsection{Irregular Uniform-Length Orthogonally Convex Polygons}

\ifFull
From now on, we assume that $P$ is an IUP, and that $G_P$ is its visibility graph.

First note that for the irregular case, no clique contains four convex vertices. Therefore, finding the tabs requires a slightly different strategy.
\else
Let $G_P$ be the visibility graph of IUP $P$. 
Our reconstruction algorithm first computes the four tabs, then assigns the convex and reflex vertices to each staircase. The following structural lemma helps us find the tabs. We assume that we have already computed the convex and reflex vertices in $O(n^2m)$ time.
\fi

\begin{lemma}
\label{lemma:cliques-convex-vertices}
In every IUP there are exactly four $7$-vertex maximal cliques, each containing exactly three convex vertices. Each such clique contains exactly one tab, and each tab is contained in exactly one of these cliques.
\end{lemma}
\begin{proof}
First note that each of the four tabs are in exactly one such maximal $7$-clique. Further, any other clique that contains three convex vertices has at least nine vertices: each convex vertex and its two reflex boundary neighbors.
\end{proof}

\ifFull
\begin{lemma}
\label{lemma:8-edges}
In an IUP, we can find a set of at most eight edges that contains the four tabs edges in $O(n^2m)$ time.
\end{lemma}
\begin{proof}
Compute the four maximal cliques from Lemma~\ref{lemma:cliques-convex-vertices}. These cliques have exactly three convex vertices each, and tabs are incident to two convex vertices, narrowing our choice of tab down to $4\cdot\binom{3}{2} = 12$ edges. Four of these are vertical or horizontal (non-boundary) edges. Using an argument similar to Lemma~\ref{lemma:convex-neighbor-two-cliques} (Case 3), it is easy to show that they are not $1$-simplicial, and therefore, can be identified and eliminated. Thus, we have eight edges to consider.
\end{proof}

\ifFull
Note that each clique must contain a tab.

Furthermore, the edges in the irregular case form 2-paths. The middle vertex on each 2-path is a tab vertex.
\else
Furthermore these eight edges form four disjoint paths on two edges, and the middle vertex on each path is a tab vertex. Lastly, we note that these known tab vertices are on the long staircases.
\fi
\fi

\ifFull
\else
We note that it is not necessary to identify the four tabs explicitly to continue with the reconstruction. There are only $7^4=O(1)$ choices of tabs (one from each $7$-clique of Lemma~\ref{lemma:cliques-convex-vertices}), thus we can try all possible tab assignments, continue with the reconstruction and verify that our reconstruction produces a valid IUP $P$ with the same visibility graph.
However, we can explicitly find the four tabs, \ifAppendix which we show in Appendix~\ref{subsection:find-4-tabs},\else \fi{} giving us the following lemma.
\fi

\ifFull
We now show how to eliminate the remaining four non-tabs.
\fi

\begin{lemma}
\label{lemma:iup-tabs}
We can identify the four tabs of an IUP in $O(nm)$ time.
\end{lemma}
\ifFull
\begin{proof}
We identify the four maximal cliques from Lemma~\ref{lemma:cliques-convex-vertices}. These cliques have exactly three convex vertices each, and tabs are incident to two convex vertices, narrowing our choice of tab down to $4\cdot\binom{3}{2} = 12$ edges. Four of these are vertical or horizontal (non-boundary) edges. Using argument similar to Case 3 of Lemma~\ref{lemma:convex-neighbor-two-cliques}, we can show that they are not $1$-simplicial. Thus, we can detect and eliminate them. 

The remaining eight edges form four disjoint paths on two edges, and the middle vertex on each path is a tab vertex. Note that these middle tab vertices are on the long staircases. Let one of them be called $u$. Now it remains to find $u$'s neighbor on its tab. Vertex $u$ has two candidate neighbors; let's call them $v$ and $w$. Just for concreteness, let's say $u$ is the vertex of the north tab on the northeast staircase.

Suppose, without loss of generality, that $v$ has more reflex neighbors than $w$, then $v$ is $u$'s neighbor on a tab, because it sees reflex vertices on the whole northeast and southeast staircases, while $w$ sees only a subset of those.
Otherwise $v$ and $w$ have the same number of reflex neighbors. Then either $v$ or $w$ has more convex neighbors. Suppose, without loss of generality, that $v$ is a tab vertex, then $v$ has fewer convex neighbors than $w$. To see why, note that since $u$ is on the northeast (long) staircase, $v$ is on the northwest (short) staircase. Vertex $v$ has convex neighbors $u$, $w$, and every convex vertex on the southeast (short) staircase. Likewise, $w$ has convex neighbors $v$, $u$, every convex vertex on the northeast (long) staircase (including $u$) and one convex vertex on the southeast (short)

We can do these checks for all such pairs $v$ and $w$, giving us all tabs.
\end{proof}
\fi

\ifFull
\begin{figure}[!tbh]
\begin{center}
\includegraphics[scale=0.8]{figures/elementary-cliques-smaller}
\caption{Elementary cliques of an IUP. Only half of the elementary cliques are shown.}
\label{figure:elementary-cliques}
\end{center}
\end{figure}
\else
\begin{figure}[!tb]
\begin{center}
\subfloat[]{\includegraphics[scale=0.55]{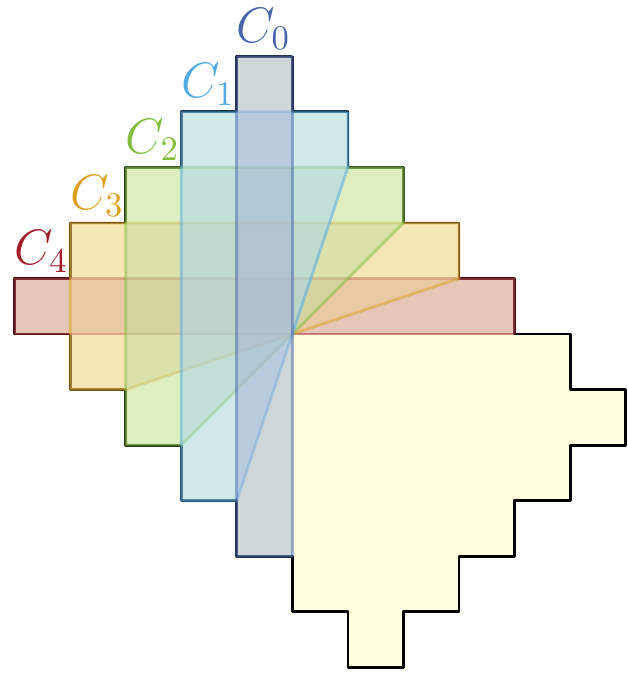}\label{cliques}}\hspace{8pt}
\subfloat[]{\includegraphics[scale=0.55]{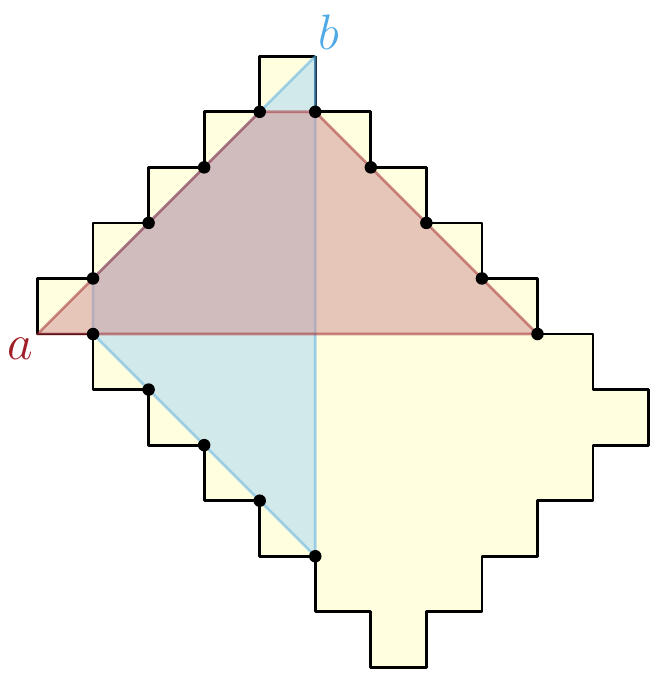}\label{ab-visibility}}\hspace{8pt}
\subfloat[]{\includegraphics[scale=0.55]{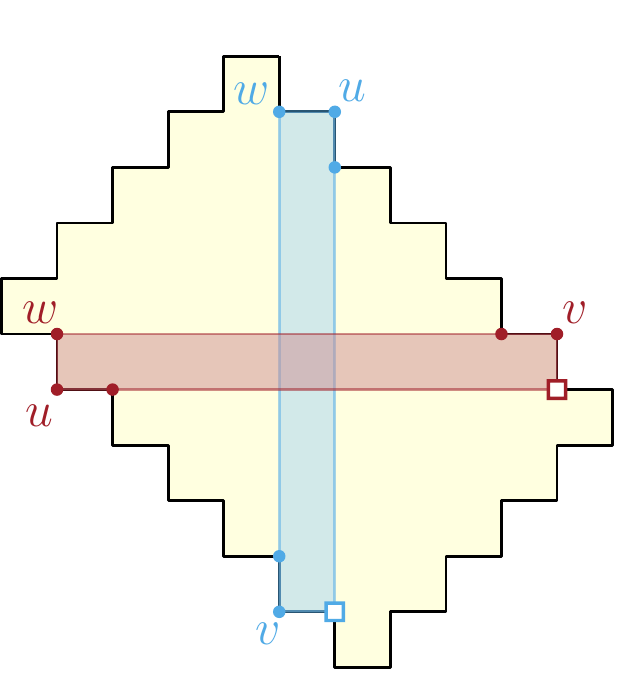}\label{rectangles}}
\caption{Elements of our reconstruction. (a) Elementary cliques $C_0, \dots, C_4$ interlock along a short staircase. (b) Tab vertices $a$ and $b$ see unique reflex vertices on long staircases. (c) Reflex vertices (square) are discovered by forming rectangles with known vertices $u$, $v$ and $w$.}
\label{figure:algorithm-stages}
\end{center}
\end{figure}
\fi

We pick one tab arbitrarily to be the north tab. 
%
We conceptually orient the polygon so that the northwest staircase is short and the northeast staircase is long. We do this by computing \emph{elementary cliques}, which identify the convex vertices on the short staircase.

\begin{definition}[elementary clique] An \emph{elementary clique} in an IUP is a maximal clique that contains exactly three convex vertices: one from a short staircase, and one from each of the long staircases. \ifFull (See \fig~\ref{figure:elementary-cliques}.)\else (See \fig~\ref{figure:algorithm-stages}\subref{cliques}.)\fi
\end{definition}

\begin{lemma}
\label{lemma:elementary-cliques}
We can identify the elementary cliques containing vertices on the northwest staircase in $O(nm)$ time.
\end{lemma}
\ifFull
\begin{proof}
Each elementary clique $C$ has constant size and contains a $1$-simplicial edge, as $C$ is maximal and two convex vertices in $C$ must be on opposite staircases. We therefore compute $1$-simplicial edges and their maximal cliques, and keep the cliques that are elementary cliques.

Let $k_{NW}$ be the number of convex vertices on the northwest staircase, and number these convex vertices from $v_0$ to $v_{k_{NW}-1}$ in order along the northwest staircase from the north tab to the west tab. We denote by $C_i$ the unique elementary clique containing $v_i$. Note that $C_0$ is the unique maximal (elementary) clique containing the north tab. Furthermore, each clique $C_i$ contains a set of reflex vertices $R_i$ such that $|R_i\cap C_{i+1}| = 3$, and for $j \not\in\{i-1,i,i+1\}$, $R_i\cap C_j = \emptyset$.
Therefore, from elementary clique $C_i$, we can compute elementary clique $C_{i+1}$ by searching for the only other elementary clique containing reflex vertices $R_i$. Once we reach an elementary clique containing a tab, then we have computed all elementary cliques on the northwest staircase. This tab is the west tab and we are finished.
\end{proof}
\else
\begin{proof}[\sketch]
Each elementary clique is constant size and contains a $1$-simplicial edge, and can therefore be discovered in $O(nm)$ time. Further, elementary cliques ``interlock'' along a staircase: each elementary clique shares exactly three reflex vertices with its at most two neighboring elementary cliques. Thus they can be computed starting from the elementary clique containing the north tab. \ifAppendix See Appendix~\ref{subsection:elementary-cliques} for the full proof.\fi{}
\end{proof}
\fi

Note that, if our sole purpose is to reconstruct the IUP $P$, we have sufficient information. The number of elementary cliques gives us the number of vertices on a short staircase of the polygon, from which we can build a polygon. However, in what follows, we can actually map all vertices to their positions in the IUP, which we later use to build a recognition algorithm for IUPs.

\ifFull
From the elementary cliques, we can determine all the convex vertices on the northwest side, and further divide the convex vertices from the elementary cliques to the southwest and northeast sides.
\else
First, we show how to assign all convex vertices from the elementary cliques to each of the three staircases, using visibility of the north and west tab vertices. Note, constructing the elementary cliques with Lemma~\ref{lemma:elementary-cliques} also gives us the west tab, since it is contained in the last elementary clique on the northwest staircase.
\fi

\begin{lemma}
\label{lemma:cliques}
We can identify the convex vertices on the northwest staircase in $O(n)$ time.
\end{lemma}
\begin{proof}
The northwest staircase contains the convex vertices of the elementary cliques from Lemma~\ref{lemma:elementary-cliques} that cannot be seen by any of the north or west tab vertices. The staircase further contains the left vertex of the north tab and the top vertex of the west tab (which can be identified by the fact that they are tab vertices that do not see either vertex of the other tab). 
\end{proof}


We can repeat the above process to identify the convex vertices of the southeast staircase. However, we might not yet be able to identify tabs as south or east. Thus, we will obtain two possible orderings of the convex vertices on the southeast staircase.
Next, we show how to assign convex vertices to the long staircases. In the process we determine south and east tabs, and consequently, identify the correct ordering of convex vertices on the southeast staircase.

\begin{lemma}
\label{lemma:long-staircases}
We can assign the remaining convex vertices 
in $O(n^2)$ time.
\end{lemma}
\ifFull
\newcommand{\W}{\mathcal{W}}
\newcommand{\E}{\mathcal{E}}
\begin{proof}
Let $\W$ and $\E$ be all the convex vertices on the southwest and northeast staircases and let $\W_0$ and $\E_0$ be the convex vertices  on the southwest and northeast staircases that are already known from the elementary cliques from Lemma~\ref{lemma:cliques}. Let $N_c(v)$ denote the set of convex neighbors on the opposite staircase of some vertex $v$. Then, for each vertex $w_0\in \W_0$, $N_c(w_0) \subseteq \E$. 
Similarly, for each vertex $e_0\in \E_0$, $N_c(e_0)\subseteq \W$. Then we can iteratively define sets $\E_i = (\cup_{w\in \W_{i-1}} N_c(w)) \setminus \E_{i-1}$ and $\W_i =(\cup_{e\in \E_{i-1}} N_c(e)) \setminus \W_{i-1}$ and identify all convex vertices of the southwest and northeast staircases as $\W = \cup_i \W_i$ and $\E = \cup_i \E_i$. 

To order the vertices along the southwest staircase, note that the sets $\W_i$ should appear in order of increasing $i$ from top to bottom. Also note that if one were to assign the vertices of $W_i$ the staircase from top to bottom, each vertex $w_i$ in this order would see fewer vertices of $\E_{i-1}$. Thus, we can order the vertices within each $\W_i$. The argument for ordering vertices of $\E_i$ is symmetric.

Finally, the east (resp., south) tab is identified by the fact that its top (resp., left) vertex is visible from at least one vertex in $\W$ (resp, $\E$) that is not a tab vertex.
\end{proof}
\else
\begin{proof}[\sketch]
Let $v_i,v_{i+1}$ be convex vertices on the same staircase, separated by a single reflex vertex. Let $u$ be the unique vertex on the opposite staircase, such that the angular bisector of $v_i$ goes through $u$. Then $u$ sees $v_{i+1}$. This likewise holds for the opposite staircase. 
Therefore, starting from one convex vertex on each staircase (such as a tab vertex), we can compute all convex vertices on each staircase.
\end{proof}
\fi

\ifFull
We can now choose the south and east edges: a vertex on the east (south) edge can see convex vertices on the southwest (northeast) staircase. 
\fi

\ifFull
\begin{figure}[!tbh]
\begin{center}
\includegraphics{figures/side-visibility-small}
\caption{Left: Tab vertices $a$ and $b$ see unique reflex vertices on long staircases. Right: We assign the remaining reflex vertices with rectangles between long staircases.}
\label{figure:side-visibility}
\end{center}
\end{figure}
\fi

\begin{lemma}
\label{lemma:remaining-reflex}
We can assign the reflex vertices to each staircase in $O(n^2)$ time.
\end{lemma}
\ifFull
\begin{proof}
Once the convex vertices are ordered on the staircases, we can compare the reflex vertices that are seen from the tab vertices. Let $a$ and $b$ be vertices on different tabs visible along a short staircase \ifFull(see \fig~\ref{figure:side-visibility}(left))\else(see \fig~\ref{figure:algorithm-stages}\subref{ab-visibility})\fi. Let $\R$ be the set of all reflex vertices of the IUP, and $N(v)$ be a set of all neighbors of vertex $v$ in the visibility graph of IUP. Then $\R_0 = N(a)\cap N(b)\cap \R$ contains all reflex vertices from the short staircase, plus two extra reflex vertices from the neighboring long staircases. The remaining vertices $N(a)\setminus \R_0$ are on one long staircase (and $N(b)\setminus R_0$ are on the other long staircase).

Thus, we can find many reflex vertices on the long staircases, except the endvertices and potentially those in the middle of the staircases. To find the remaining ones, we build rectangles (maximal cliques) consisting of two convex vertices $u$ and $v$ on the opposite staircases and a known reflex vertex $w$, such that $(u,w)$ forms a boundary edge of the IUP \ifFull(see \fig~\ref{figure:side-visibility}(right))\else(see \fig~\ref{figure:algorithm-stages}\subref{rectangles})\fi. These rectangles define new reflex vertices on the opposite staircase from $w$. Thus, we iteratively discover all new reflex vertices.
\end{proof}
\else
\begin{proof}[\sketch] 
First we compare the reflex vertices seen by tab vertices, which gives us many vertices on the long staircases (see \fig~\ref{figure:algorithm-stages}\subref{ab-visibility}). The remaining reflex vertices are discovered by building vertical and horizontal rectangles that contain unassigned reflex vertices (see \fig~\ref{figure:algorithm-stages}\subref{rectangles}). \ifAppendix See Appendix~\ref{subsection:remaining-reflex} for the full proof.\fi{}
\end{proof}
\fi{}

\ifFull
\begin{lemma}
We can order the reflex vertices on each staircase in $O(n^2m)$ time.
\end{lemma}
\begin{proof}
\ifFull We already have ordered the convex vertices on each staircase.\fi{} Let $c_0,\ldots,c_k$ be the convex vertices in order on some staircase $S$ containing reflex vertices $R$. Then $N(c_i)\cap N(c_{i+1})\cap R=\{r_i\}$ where $r_i$ is the reflex vertex between $c_i$ and $c_{i+1}$ on staircase $S$. Thus, we know the order of the reflex vertices along each staircase.
\end{proof}

Therefore, we have ordered the vertices on all boundary staircases, constructing the Hamiltonian cycle of boundary edges in $G_P$, arriving at the following theorem:
\else
Observe that within each staircase, boundary edges are formed only between convex vertices and their reflex neighbors. Thus, we can order reflex vertices on each staircase by iterating over the staircase's convex vertices (order of which is determined in Lemmas~\ref{lemma:elementary-cliques}-\ref{lemma:long-staircases}) and we are done. This gives us the following result:

\fi

\begin{theorem}
In $O(n^2m)$ time, we can reconstruct an IUP from its visibility graph.
\end{theorem}


\newcommand{\hpolygon}{histogram}
\newcommand{\Hpolygon}{Histogram}

\section{Uniform-Length \Hpolygon~Polygons}\label{sec:monotone}
\ifFull
\begin{definition}[\hpolygon~polygon]
Let $P$ be an orthogonal polygon consisting of $n-1$ length $l$ boundary edges that alternate horizontally and vertically, and one edge of length $l(n-3)/2$. Then we call $P$ a \hpolygon~polygon.
\end{definition}

The visibility graph of $P$ does not depend on the length $l$, so for the ease of discussion we treat $l=1$.
As with the double staircase polygons, we call the long edge the \emph{base} edge. Note that $P$ is monotone with respect to the base edge, as the intersection between a sweep line normal to the base edge and $P$ is always connected. However, note that the sweep line intersects the edges normal to the base edge at some point during the sweep.
\fi

\ifRemove
\begin{figure}[!tb]
\begin{center}
\subfloat[]{\includegraphics[width=0.4\textwidth]{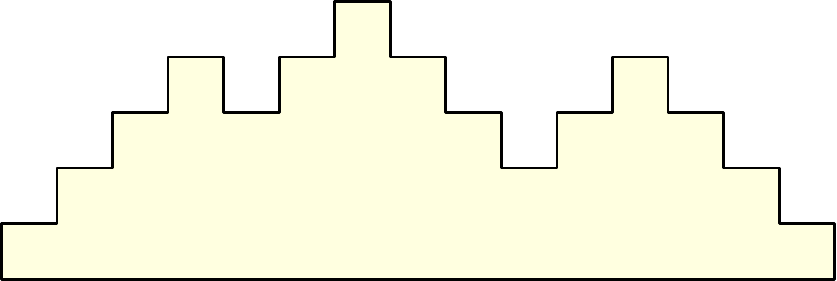}\label{monotone}}\hspace{0.8cm}
\subfloat[]{\includegraphics[width=0.4\textwidth]{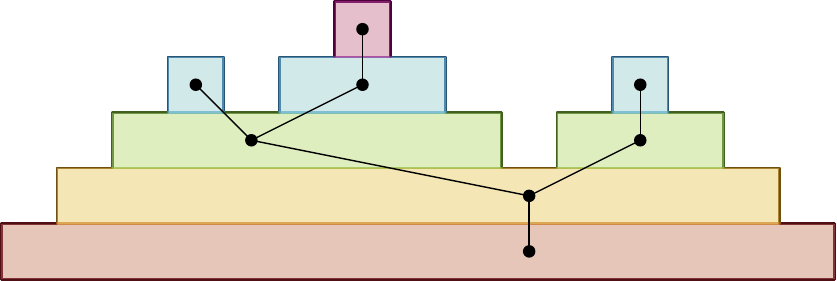}\label{tree}}
\caption{(a) A \hpolygon~$A$ with three tabs. (b) A decomposition of $A$ into touching rectangles with a contact graph that is a tree.}
\label{figure:monotone-staircase-and-contact-tree}
\end{center}
\end{figure}
\fi 

In this section we show how to reconstruct
a more general class of uniform step length polygons: those that consist of a chain of alternating up- and down-staircases with uniform step length, which are monotone with respect to a single (longer) {\em base edge}. Such polygons are uniform-length \emph{\hpolygon~polygons}~\cite{durocher-2012}, but we simply call them \emph{\hpolygon s} for brevity (see \fig~\ref{figure:monotone-staircase-and-contact-tree}\subref{monotone} for an example).
%
%
We refer to the two convex vertices comprising the base edge as \emph{base vertices}. Furthermore, we refer to top horizontal boundary edges incident to two convex vertices as \emph{tab edges} or just \emph{tabs} and their incident vertices as \emph{tab vertices}. 
 


\paragraph{The case of two staircases.}
We first note that in {\emph{double staircase} polygons (consisting of only two staircases) there is a simple linear-time reconstruction algorithm based on the degrees of vertices in the visibility graph. \ifAppendix It is an extension of the algorithm of Abello and E\u{g}ecio\u{g}lu~\cite{abello-uniform-step-length} and can be found in Appendix~\ref{subsection:double-staircase}.\else \fi{}
\ifFull
First note that the number of vertices in the polygon is $4k+4$ for some integer $k$. Then the vertices along the left staircase from left to right are $3k+3$, $5$, $3k+3$, $6$, $3k+2$, $7$, $3k+1$, $8$, $\ldots$, $k+2$, $2k+7$, $k+3$, $2k+6$, $k+4$, $2k+5$, $k+3$. Degrees $k+3$ and $3k+3$ are the only duplicates, however the tab vertex with degree $k+3$ can be distinguished as it has two neighbors with degree $2k+5$ and the bottom most staircase vertex with degree $3k+3$ has two neighbors of degree $5$. The vertices from right to left on the right staircase have the same degrees as the left staircase, and vertices may be assigned to left and right staircases through their visibility of tab and base vertices.
\fi
However, the construction relies on the symmetry of the two staircases 
and it is not clear whether any counting strategy works for arbitrary \hpolygon s.

\subsection{Overview of the Algorithm}
\ifFull
Each \hpolygon~can be decomposed into axis-aligned rectangles, as illustrated in \fig~\ref{figure:monotone-staircase-and-contact-tree}\subref{tree}.  The contact graph of this decomposition is a tree $T$ with a specific left-to-right order of the leaves, where each of the $k$ leaves corresponds to each of $k$ square regions next to the tabs. Conversely, a tree with left-to-right order of the leaves defines a valid \hpolygon. In \sec~\ref{subsection:decomposition}, we show that we can construct the (unordered) contact tree $T$ from the visibility graph $G_P$ in $O(n^2m)$ time. 
Let $\pordered$ be a monotone polygon corresponding to the $i$-th ordering of leaves of $T$ (out of $k!$ possible). Then for each $\pordered$ we  construct a visibility graph $G_{\pordered}$  and check if $G_P$ and $G_{\pordered}$ are isomorphic. The construction of $G_{\pordered}$ from $\pordered$ takes $O(n\log n + m)$ time using the algorithm of Ghosh and Mount~\cite{ghosh-mount-optimal-visibility} and the best pseudopolynomial graph isomorphism solution due to Babai~\cite{babai-2016} takes $\exp((\log n)^{O(1)})$ time. Thus the overall runtime of this solution is $O(n^2m + k!(n\log n+m + \exp((\log n)^{O(1)})))$.

However, in \sec~\ref{subsection:vertex-staircase-assignment} we show how to find a one-to-one mapping between the vertices of $G_P$ and the vertices of $\pordered$. This, in turn, provides one-to-one correspondence between the vertices of $G_P$ and $G_{\pordered}$, eliminating the expensive graph isomorphism step. 

For each of $k!$ orderings of the $k$ leaves of the contact tree $T$, the mapping algorithm in \sec~\ref{subsection:vertex-staircase-assignment} requires that all $2^k$ possible left-right orientations of the $k$ tab edges be considered. Each of $k!2^k$ possible combinations of leaf ordering and tab orientations results in a different polygon that is a potential candidate for a reconstruction of $G_P$. The mapping algorithm either outputs a valid mapping of the vertices of $G_P$ to the vertices of the candidate polygon or detects that the given candidate does not produce a valid reconstruction of $P$. Since the mapping between the vertices is provided, the final equivalence check between $G_{P}$ and the visibility graph of the candidate polygon takes only $O(n + m)$ time. In \sec~\ref{subsection:reducing-candidates} we further reduce the number of candidates on which to run the equivalence check to $(k-2)!2^{k-2}$, leading to the main result of our paper:
\else
Every \hpolygon~can be decomposed into axis-aligned rectangles, whose contact graph is an ordered tree~\cite{durocher-2012}, as illustrated in \fig~\ref{figure:monotone-staircase-and-contact-tree}\subref{tree}. In \sec~\ref{subsection:decomposition}, we show that we can construct the (unordered) contact tree $T$ from the visibility graph $G_P$ in $O(n^2m)$ time by repeatedly ``peeling'' tabs from the \hpolygon. We then show that each left-to-right ordering of $T$'s $k$ leaves (as well as a left-to-right orientation of the rectangles in the leaves) induces a \hpolygon~$P'$. For each candidate polygon $P'$ (of $k!2^k$ candidates), we then compute its visibility graph $G_{P'}$ in $O(n\log n + m)$ time~\cite{ghosh-mount-optimal-visibility} and check if $G_{P'}$ is isomorphic to $G_P$. Instead of requiring an expensive graph isomorphism check~\cite{babai-2016}, we show how to use the ordering of $T$ to quickly test if $G_P$ and $G_{P'}$ are isomorphic.
In \sec~\ref{subsection:reducing-candidates} we show how to reduce the number of candidate \hpolygon s from $k!2^k$ to $(k-2)!2^{k-2}$, leading to the main result of our paper:
\fi{}

\begin{theorem}
\label{theorem:main}
Given a visibility graph $G_P$ of a \hpolygon~$P$ with $k\ge2$ tabs, we can reconstruct $P$ in $O(n^2m + (k-2)!2^{k-2}(n\log n + m))$ time.
\end{theorem}

Finally, we give a faster reconstruction algorithm when the \hpolygon~has a binary contact tree, solving these instances in $O(n^2m)$ time (\sec~\ref{subsection:reducing-candidates}).

\newif\ifNodariNew
\NodariNewtrue

\ifNodariNew
\else
\subsection{Algorithm for \Hpolygon s~with $k \ge 2$ Tabs}
To avoid running the generic graph isomorphism solution of Babai~\cite{babai-2016} on $G_P$ and $G_{\pordered}$, for a given $\pordered$ we map every vertex of $G_P$ to the vertex of $\pordered$. This provides one-to-one correspondence between the vertices of $G_P$ and $G_{\pordered}$. Then the check if $G_{P}$ is equivalent to $G_{\pordered}$ takes only $O(n\log n + m)$ time using the algorithm of Ghosh and Mount~\cite{ghosh-mount-optimal-visibility}.

Given an ordered contact tree $T$ that corresponds to the \hpolygon~$\pordered$ and a left-right orientation of each tab edge of $\pordered$, the algorithm in \sec~\ref{subsection:vertex-staircase-assignment} either performs the mapping of vertices of $G_P$ to $\pordered$ or determines that $\pordered$ is not a valid reconstruction of $P$. Thus, there are $k!2^k$ candidates that we need to consider. In \sec~\ref{subsection:reducing-candidates} we further reduce the number of candidates to $(k-1)!2^{k-1}$, leading to the main result of our paper:

\begin{theorem}
\label{theorem:main}
Given a visibility graph $G_P$ of a \hpolygon~$P$ with $k$ tabs, we can reconstruct $P$ in $O(n^2m + (k-2)!2^{k-2}(n\log n + m))$ time.
\end{theorem}

  associates four vertices of $G_P$ with every rectangle $R_v$ in the decomposition (corresponds to a node $v$ in the contact tree $T$). Moreover, it classifies vertices as \topVs and \bottomVs of the rectangle. Given a \topV (resp., \bottomV) $p$, we say the other \topV (resp., \bottomV) in the rectangle, denoted by $\bar{p}$, is its {\em \companion} vertex. Since left and right vertices of each rectangle belong to a left (ascending) and down (descending) staircases, respectively, each pair of companion vertices $p$ and $\bar{p}$ needs and right vertices 

In \sec~\ref{subsection:decomposition} we show how to construct the contact tree $T$ from $G_P$ in $O(n^2m)$ time. (Note that such a decomposition defines the structure (length) of each staircase and needs to be performed only once.) Moreover, it associates each node $v$ of $T$ with with four vertices of $P$: two {\em \topVs} that are convex and two {\em \bottomVs} that are either both reflex or are both convex {\em base} vertices. Given a \topV (resp., \bottomV) $p$, we say the other \topV (resp., \bottomV) in the rectangle, denoted by $\bar{p}$, is its {\em \companion} vertex. Our decomposition also successfully classifies vertices to \topVs and \bottomVs. Thus, our decomposition identifies a pair of \companion vertices $p$ and $\bar{p}$ as candidates for each vertex of a staircase, with an additional constraint that the assignment of a vertex to the left (resp., right) staircase implies the assignment of its \companion vertex to the right (resp., left) staircase defining the other side of the rectangular region. Section~\ref{subsection:vertex-staircase-assignment} describes how to determine the actual assignment of vertices of $G_P$ to staircases in $O(n\log n +m)$ time. Thus, we get the following result:

We associate each ascending or descending staircase with a tab vertex at its top. Using $G_P$, we determine the lengths of each such staircase. Each ordering of these staircases, combined with alternating the ascent/descent between neighbor in the ordering defines a polygonal chain $P^\prime$. Note, that not all polygonal chains define a valid \hpolygon~(see \fig~\ref{figure:monotone-orderings} for some examples). To determine if $P^\prime$  is a valid \hpolygon, it is sufficient to check that the first and the last vertices of $P^\prime$ have the same $y$ coordinate (the base edge is horizontal) and the rest of the vertices are above the base edge. Thus, in $O(n)$ time we can verify if $P^\prime$ forms a \hpolygon. 
\begin{figure}
\begin{center}
\includegraphics[width=\textwidth]{figures/monotone-orderings}
\caption{Possible orderings of tabs $a$,$b$, and $c$ and assignments of their (red and blue) staircases. (a): An ordering that forms a \hpolygon. (b): Reordering tabs $b$ and $c$ from (a) creates a polygonal chain that intersects the base edge. (c): Beginning from the ordering in (b), reassigning $c$'s staircases creates a polygonal chain such that first and last \bottomV have different $y$-coordinates.}
\label{figure:monotone-orderings}
\end{center}
\end{figure}

Given a candidate \hpolygon~$P^\prime$, one strategy to verify if $G_P$ is a valid visibility graph  of $P^\prime$ is to construct the visibility graph $G_{P^\prime}$ of $P^\prime$ in time $O(n\log n +m)$ and check if it is isomorphic to $G_P$ in $\exp((\log n)^{O(1)})$ time using the quasipolynomial time graph isomorphism algorithm due to Babai~\cite{babai-2016}.  The total runtime is $O(n^2m + k!(n\log n+m + \exp((\log n)^{O(1)})))$.

Instead,  we develop a fixed-parameter tractable algorithm, parameterized by $k$ -- the number of tabs in $P$, obtaining time $O(n^2m + k!2^k(n\log n + m))$, which is faster for small values of $k$.

Finally, for each valid \hpolygon~$P^\prime$, we check if the visibility graph $G_{P^\prime}$ is equal to $G_P$. Our construction of the staircases from $G_P$ defines one-to-one correspondence between vertices of $G_P$ and the vertices of $P'$. Therefore, the equivalence check of $G_P$ with $G_{P^\prime}$ requires $O(n\log n +m)$ time using the algorithm of Ghosh and Mount~\cite{ghosh-mount-optimal-visibility}.

The number of candidates for $P^\prime$ that we consider will be parameterized by $k$ -- the number of tabs in $P$. Observe that for every vertical boundary edge $e$ in $P$ that belongs to some staircase, there is a {\em \companion} boundary edge $\bar{e}$ on another staircase. The four vertices of $e$ and $\bar{e}$ define axis-aligned rectangular region $R$.  The contact graph of the decomposition of a polygon into such rectangular regions is a tree $T$ (see \fig~\ref{figure:monotone-staircase-and-contact-tree}(right) for an example). We root $T$ at the rectangle containing the base edge (the base rectangle); consequently,  each leaf is a rectangle containing a tab (a tab rectangle). 
The left-to-right ordering of the leaves of $T$ defines the left-to-right ordering of the tabs of the polygon, while the left-to-right orientation of the tab edges defines the left-right assignment of the two staircases associated with each tab. Thus, there are $k!2^k$ potential candidates for $P^\prime$.

In \sec~\ref{subsection:decomposition} we show how to construct the contact tree $T$ from $G_P$ in $O(n^2m)$ time. (Note that such a decomposition defines the structure (length) of each staircase and needs to be performed only once.) Moreover, it associates each node $v$ of $T$ with with four vertices of $P$: two {\em \topVs} that are convex and two {\em \bottomVs} that are either both reflex or are both convex {\em base} vertices. Given a \topV (resp., \bottomV) $p$, we say the other \topV (resp., \bottomV) in the rectangle, denoted by $\bar{p}$, is its {\em \companion} vertex. Our decomposition also successfully classifies vertices to \topVs and \bottomVs. Thus, our decomposition identifies a pair of \companion vertices $p$ and $\bar{p}$ as candidates for each vertex of a staircase, with an additional constraint that the assignment of a vertex to the left (resp., right) staircase implies the assignment of its \companion vertex to the right (resp., left) staircase defining the other side of the rectangular region. Section~\ref{subsection:vertex-staircase-assignment} describes how to determine the actual assignment of vertices of $G_P$ to staircases in $O(n\log n +m)$ time. Thus, we get the following result:

\begin{lemma}
Reconstructing \hpolygon s with $k$ tabs is fixed-parameter tractable, with running time $O(n^2m + k!2^{k}(n\log n + m))$.
\end{lemma}

TODO: Mention improvements to special cases of polygons in sections 4.3 and 4.4.
\fi 

\subsection{Rectangular Decomposition and Contact Tree Construction}
\label{subsection:decomposition}
We construct the contact tree $T$ from $G_P$ by computing a set $\mathcal{T}$ of the $k$ tab edges of $G_P$ (Lemma~\ref{lemma:monotone-tab-id}). 
%
Each tab $(u,v)$ is $1$-simplicial and in a maximal $4$-clique, since $N(u)\cap N(v)$ is a $4$-clique representing a unit square at the top of the \hpolygon. Given the set $\mathcal{T}$ of tab edges, our reconstruction algorithm picks an edge $t$ from $\mathcal{T}$ and removes the maximal $4$-clique containing $t$. This is equivalent to removing an axis-aligned rectangle in $P$, and, equivalently, removing a leaf node from $T$. Moreover, it associates that node of $T$ with four vertices of $P$: two {\em \topVs} that are convex and two {\em \bottomVs} that are either both reflex or are both convex {\em base} vertices. This process might result in a new tab edge, which we identify and add to $\mathcal{T}$. 

\ifLipics
\subsubsection{Finding initial tabs}
\else
\subsubsection{Finding initial tabs.}
\fi 
We start by finding the $k$ tabs. Recall that every tab edge is $1$-simplicial and in a maximal $4$-clique. The converse is not necessarily true. Therefore, we begin by finding all $1$-simplicial edges that are in maximal $4$-cliques as a set of candidate edges, and later exclude non-tabs from the candidates. 

\ifFull
\begin{figure}[!tb]
\begin{center}
\includegraphics[width=0.9\textwidth]{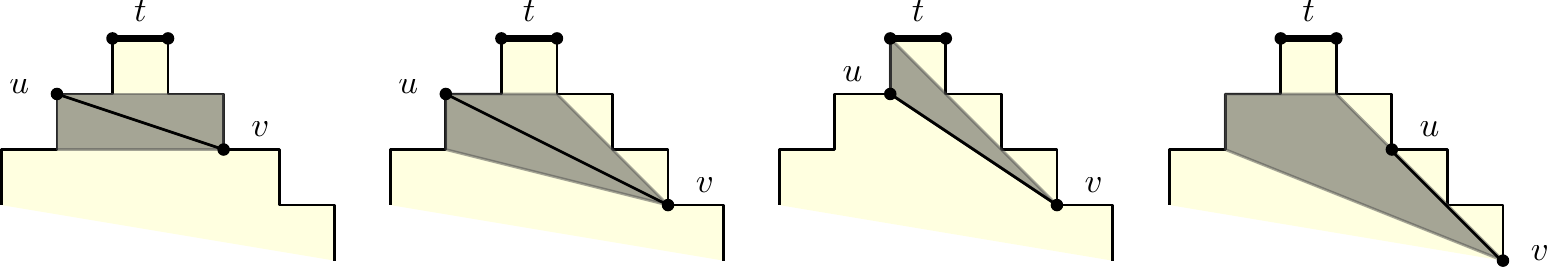}
\caption{Edges between $u,v$ on a tab $t$'s staircases, but disjoint from $t$, are in a clique of size at least five.}
\label{figure:cliques-in-tab-staircases}
\end{center}
\end{figure}
\fi



\newcommand{\iso}{isolated vertex\xspace}
\newcommand{\isos}{isolated vertices\xspace}
Given a visibility graph $G_P=(V_P, E_P)$ of a \hpolygon~$P$ and a maximal clique $C\subseteq V_P$, we call a vertex $w \in C$ an {\em \iso{} with respect to $P$} if there exists a tab edge $(u,v) \in E_P$, such that $(N(u) \cup N(v)) \cap C = \{w\}$, i.e., of all vertices of $C$, only $w$ is visible to some tab of $P$.

\begin{figure}[!tb]
\begin{center}
\includegraphics[width=.6\textwidth]{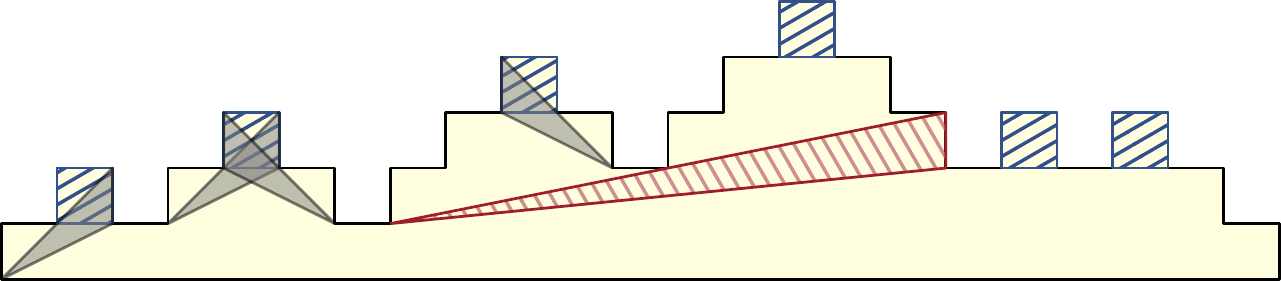}
\caption{Illustrating all maximal $4$-cliques that contain $1$-simplicial edges. These include tab cliques (square regions) and non-tab cliques (triangular regions).}
\label{figure:monotone-staircase-tab-candidates}
\end{center}
\end{figure}

\begin{lemma}
\label{lemma:isolated-vertices}
In a \hpolygon, every $1$-simplicial edge in a maximal $4$-clique contains either a tab vertex or an \iso. 
\end{lemma}
\ifSketch
\begin{proof}[\sketch]
Figure~\ref{figure:monotone-staircase-tab-candidates} shows the only types of maximal $4$-cliques. \ifAppendix Formal proof is in Appendix~\ref{subsection:islated-vertices}.\fi{}
\end{proof}
\fi 

\ifFull

We now show how to compute the tabs from the $1$-simplicial edges.
\fi

\newcommand{\test}{\ensuremath{V_{12}}}   
\newcommand{\csim}{\C_{\mathrm{sim}}}

\ifFull
\begin{figure}[!tb]
\begin{center}
\includegraphics[width=.8\textwidth]{figures/exclude-support-cliques}
\caption{The possible configurations of tab cliques (red) and staircase support cliques (gray) of size $4$.} 
\label{figure:exclude-support-cliques}
\end{center}
\end{figure}
\fi

\newcommand{\esim}{E_{\mathrm{sim}}}
\begin{lemma}\label{lemma:monotone-tab-id}
In a visibility graph of a \hpolygon, tabs can be computed in time $O(n^2m)$.
\end{lemma}
\ifSketch
\begin{proof}[\sketch]
We find all maximal 4-cliques in $O(nm)$ by Observation~\ref{claim:simplicial-k-clique-time} and detect and eliminate those containing \isos in $O(n^2m)$ time. \ifAppendix For details see Appendix~\ref{subsection:monotone-tab-id}.\fi{}
\end{proof}
\fi{}

Note that \topVs cannot see the vertices above them. Therefore, only \bottomVs see tab vertices. Moreover, every \bottomV sees at least one tab vertex. Thus, identifying all tabs immediately classifies vertices of $G_P$ into \topVs and \bottomVs.

\ifFull
\begin{figure}[!tb]
\begin{center}
\includegraphics[width=0.9\textwidth]{figures/truncated-monotone-staircase}
\caption{Left: A \hpolygon~$A$. Right: A truncated \hpolygon, formed by starting from $A$ and iteratively removing six tab rectangles (dashed).}
\label{figure:truncated-monotone-staircase}
\end{center}
\end{figure}
\fi 

\ifLipics
\subsubsection{Peeling tabs}
\else
\subsubsection{Peeling tabs.}
\fi 
Let $P^\prime$ be a polygon resulting from peeling tab cliques (rectangles) from a \hpolygon~$P$. We call $P^\prime$ a \emph{truncated} \hpolygon. See \fig~\ref{figure:truncated-monotone-staircase-abbv}\subref{truncated} for an example.  After peeling a tab clique, the resulting polygon does not have uniform step length and the visibility graph may no longer have the properties on which Lemma~\ref{lemma:monotone-tab-id} relied to detect initial tabs. Instead, we use the following lemma to detect newly created tabs during tab peeling.  

\ifFull
\begin{figure}[!tb]
\begin{center}
\includegraphics[scale=0.6]{figures/truncated-monotone-staircase-crossing-edges}
\caption{The evolution of a crossing edge in a monotone (then truncated monotone) staircase. The crossing edge changes once it is incident to a tab vertex. As the \hpolygon~continues to be truncated, newly introduced crossing edges share the same lower vertex as disappearing crossing edges.}
\label{figure:truncated-monotone-staircase-crossing-edges}
\end{center}
\end{figure}
\fi


\begin{figure}[!tb]
\begin{center}
\subfloat[]{\includegraphics[scale=0.7]{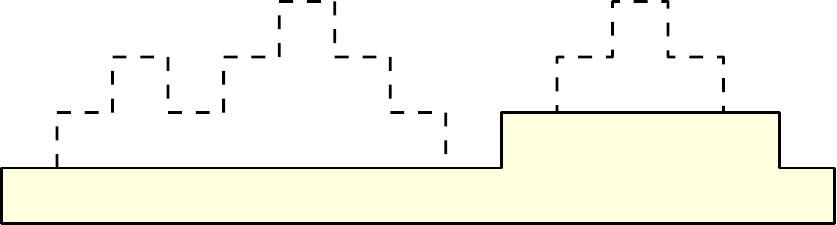}\label{truncated}}\hspace{0.8cm}
\subfloat[]{\includegraphics[scale=0.7]{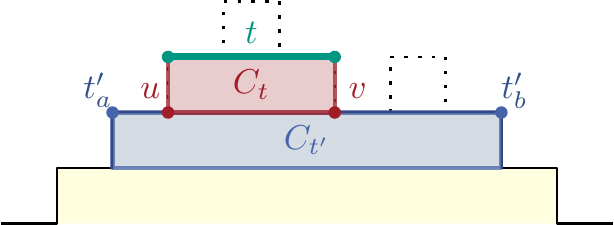}\label{tab-discovery}}
\caption{(a) A truncated \hpolygon, created by iteratively removing six tabs (dashed) from a \hpolygon. (b) When removing $C_t$: $t^{\prime}_a,t^{\prime}_b$ form a tab iff $t^{\prime}_a,t^{\prime}_b\in T\setminus C_t$, they see $u$ and $v$, and $|C_{t^\prime}|=4$.}
\label{figure:truncated-monotone-staircase-abbv}
\end{center}
\end{figure}

\begin{lemma}
\label{lemma:truncated-tab}
When removing a tab clique from the visibility graph of a truncated \hpolygon, any newly introduced tab can be computed in time $O(n)$.
\end{lemma}

\begin{proof}
Denote the removed (tab) clique by $C_t$ and let $t$ be its tab. Let $u,v \not \in t$ be the non-tab vertices of $C_t$. Since $u$ sees $v$, $(u,v)$ is an edge in $G_P$.

Since \topVs can only see vertices at and below their own level, besides the vertices of $t$, there are exactly two other \topVs in (remaining) $G_P$ that see $u$ and $v$, namely, the \topVs $t^\prime_a$ and $t^\prime_b$ of $P$ on the same level as $u$ and $v$ (see \fig~\ref{figure:truncated-monotone-staircase-abbv}\subref{tab-discovery}). 
Since $t^\prime_a, t^\prime_b$ are adjacent in $G_P$, let $t^\prime = (t^\prime_a, t^\prime_b)$.

When removing $C_t$ from $G_P$, we can compute $t^\prime_a$ and $t^\prime_b$ in time $O(n)$ by selecting the only two \topVs adjacent to both $u$ and $v$. Since $t^\prime_a$ and $t^\prime_b$ are the \topVs of a same rectangle $R_{t'}$, edge $t^\prime$ is $1$-simplicial and is in exactly one maximal clique  $C_{t^{\prime}} = N(t^\prime_a)\cap N(t^\prime_b)$, which corresponds to the convex region $R_{t'}$. Finally, after $C_t$ is removed, $t^\prime$ is a newly created tab if and only if $|C_{t^{\prime}}| = 4$, which can again be tested in time $O(n)$ by computing $N(t^\prime_a)\cap N(t^\prime_b)$. 
\end{proof}


With each tab clique (rectangle) removal, we iteratively build the parent-child relationship between the rectangles in the contact tree $T$ as follows. Using an array $A$, we maintain references to cliques being removed whose parents in $T$ have not been identified yet. 
When a tab clique $C_t$ is removed from $G_P$, the reference to $C_t$ is inserted into $A[u]$, where $u$ is one of the rectangle's bottom vertices. If the removal of $C_t$ creates a new tab $t' = (t'_a, t'_b)$, we identify $C_{t'}$  in $O(n)$ time using Lemma~\ref{lemma:truncated-tab}. Recall that $t'$ sees all bottom vertices on the same level. Thus, for every bottom vertex $u \in N(t'_a)$ (in the original graph $G_P$), if $A[u]$ is non-empty,  we set $C_{t'}$ as the parent of the clique stored in $A[u]$ and clear $A[u]$. This takes at most $O(n)$ time for each peeling of a clique. We get the following lemma, where the time is dominated by the computation of the initial tabs:


\begin{lemma}\label{lemma:decomposition}
In $O(n^2m)$ time we can construct the contact tree $T$ of $P$, associate with each $v\in T$ the four vertices that define the rectangular region of $v$, and classify vertices of $G_P$ as \topVs and \bottomVs. 
\end{lemma}


\subsection{Mapping Candidate Polygon Vertices to the Visibility Graph} 
\label{subsection:vertex-staircase-assignment}

\ifNodariNew

\newcommand{\tordered}{\ensuremath{\hat T}}
\newcommand{\pordered}{\ensuremath{\hat P}}

Let $\tordered$ correspond to $T$ with some left-to-right ordering of its leaves and let $\pordered$ be the polygon corresponding to $\tordered$. We will map the vertices of $G_P$ to the vertices of $\pordered$ by providing for each vertex of $G_P$ the $x$- and $y$-coordinates of a corresponding vertex of $\pordered$. 
Let $t_1, t_2, \dots, t_k$ be the order of the tabs in $\pordered$. Since $\tordered$ unambiguously defines the polygon $\pordered$, each node $v$ of $\tordered$ is associated with a rectangular region on the plane, and the four vertices of $G_P$ are associated with the four corners of the rectangular region. Since by Lemma~\ref{lemma:decomposition} every vertex of $G_P$ is classified as a \topV or a \bottomV, the $y$-coordinate can be assigned to all vertices unambiguously, because there are two \topVs and two \bottomVs associated with each node $v$ of $\tordered$. For every pair $p, \bar{p}$ of \topVs or \bottomVs associated with a node in $\tordered$ (we call them {\em \companion} vertices) there is a choice of two $x$-coordinates: one associated with the left boundary and one associated with the right boundary of the rectangular region.  Thus, determining the assignment of each \topV and \bottomV in $G_P$ to the left or the right boundary is equivalent to defining $x$-coordinates for all vertices in $G_P$. Although there appears to be $2^{n/2}$ possible such assignments, there are many dependencies between the assignments due to the visibility edges in $G_P$. In fact, we will show that by choosing the $x$-coordinates of the tab vertices, we can assign all the other vertices. Thus, in what follows we consider each of the $2^k$ possible assignments of $x$-coordinates to the $2k$ tab vertices.

At times we must reason about the assignment of a vertex to the left (right) staircases associated with some tab $t_j$. Given $\tordered$, the $x$-coordinates of each vertex in the left and right staircase associated with every tab $t_j$ is well-defined. 
Therefore, assigning a vertex $p$ to a left (right) staircase of some tab $t_j$ defines the $x$-coordinate of $p$. 

\newif\ifNodariRemove
\NodariRemovetrue

\ifNodariRemove
\else

In this section we assume that we have constructed the contact tree $T$ based on $G_P$ with a specific left to right ordering of the leaf nodes. We also assume that we are given a specific left-right orientation of the tab edges. We say the left tab vertex identifies the left (ascending) staircase, while the right tab vertex identifies the right (descending) staircase. Recall that the contact tree $T$ associates two \topVs and two \bottomVs with each node $u$ of $T$. Let $v$ and $w$ be the leftmost and the rightmost leaves of the subtree rooted at $u$, and let $t_v$ and $t_w$ be the tab edges in the clique associated with nodes $v$ and $w$, respectively. That means that each pair of companion vertices $p, \bar{p} \in G_P$ has a choice of being assigned to one of the two staircases: the left staircase identified by the left vertex of $t_v$ and the right staircase identified by the right vertex of $t_w$. Therefore, stating whether a vertex is assigned to a left or right staircase is sufficient to determine the $x$- and $y$-coordinates of the point on the polygon.

We associate each ascending or descending staircase with a tab vertex at its top. Using $G_P$, we determine the lengths of each such staircase. Each ordering of these staircases, combined with alternating the ascent/descent between neighbor in the ordering defines a polygonal chain $P^\prime$. Note, that not all polygonal chains define a valid \hpolygon~(see \fig~\ref{figure:monotone-orderings} for some examples). To determine if $P^\prime$  is a valid \hpolygon, it is sufficient to check that the first and the last vertices of $P^\prime$ have the same $y$ coordinate (the base edge is horizontal) and the rest of the vertices are above the base edge. Thus, in $O(n)$ time we can verify if $P^\prime$ forms a \hpolygon. 
\begin{figure}
\begin{center}
\includegraphics[width=\textwidth]{figures/monotone-orderings}
\caption{Possible orderings of tabs $a$,$b$, and $c$ and assignments of their (red and blue) staircases. (a): An ordering that forms a \hpolygon. (b): Reordering tabs $b$ and $c$ from (a) creates a polygonal chain that intersects the base edge. (c): Beginning from the ordering in (b), reassigning $c$'s staircases creates a polygonal chain such that first and last \bottomV have different $y$-coordinates.}
\label{figure:monotone-orderings}
\end{center}
\end{figure}
\fi 

In a valid \hpolygon, companion vertices $p$ and $\bar{p}$ must be assigned distinct $x$-coordinates. Therefore, after each assignment below, we check the companion vertex and if they are both assigned the same $x$-coordinate, we exclude the current polygon candidate $\pordered$ 
from further consideration. 

We further observe that in a valid \hpolygon, if a bottom vertex $p$ is not in the tab clique, then it sees exactly one tab vertex, which lies on the opposite staircase associated with that tab. Thus, we assign every such bottom vertex the left (right) $x$-coordinate if it sees the right (left) tab vertex.

Next, consider any node $v$ of the contact tree $\tordered$ and let $R_v$ define the rectangle associated with $v$ in the rectangular decomposition of a valid \hpolygon. Let $p$ be a \topV in $R_v$ and let $S(p)$  be the set of vertices visible from $p$ that are not in $R_v$ ($S(p)$ can be determined from the neighborhood of $p$ in $G_P$). Observe that if $p$ is assigned the left (right) $x$-coordinate, then every vertex in $S(p)$ is a \bottomV to the right (left) of the rectangle $R_v$, none of them belongs to a tab clique (i.e., all of them are already assigned $x$-coordinates), and all of them are assigned a right (left) $x$-coordinate. Since the $x$- and $y$-coordinates of the boundaries of $R_v$ are well-defined by $\tordered$ (regardless of vertex assignment), if $S(p)$ is non-empty, we check all of the above conditions and assign $p$ an appropriate $x$-coordinate. If a condition is violated, then the current polygon candidate is invalid and we exclude it from further consideration.

Let $p$ be one of the remaining \topVs without an assigned $x$-coordinate. If the companion $\bar{p}$ is assigned an $x$-coordinate, we assign $p$ the other choice of the $x$-coordinate. Otherwise, both $p$ and $\bar{p}$ see only the vertices inside their rectangle. In this case, the neighborhoods $N(p)$ and $N(\bar{p})$ are the same and we can assign $p$ and $\bar{p}$ to the opposite staircases arbitrarily.

Thus, the only remaining vertices without assigned $x$-coordinates are \bottomVs in tab cliques. Let $R$ be the rectangle defined by the tab and $S_{right}(p)$ (resp., $S_{left}(p)$) denote the set of vertices that $p$ sees among the vertices to the right (resp., left) of $R$. Consider a companion pair $p$ and $\bar{p}$ of \bottomVs that are in a tab clique. Observe that if $p$ is on the left boundary, then $S_{right}(\bar{p}) \subseteq S_{right}(p)$ or $S_{left}(p) \subseteq S_{left}(\bar{p})$. Symmetrically, if $\bar{p}$ is on the left boundary then $S_{right}(p) \subseteq S_{right}(\bar{p})$ or $S_{left}(\bar{p}) \subseteq S_{left}(p)$. Thus, if $|S_{right}(\bar{p})| \not = |S_{right}(p)|$ and $|S_{left}(p)| \not = |S_{left}(\bar{p})|$, we can assign $p$ and $\bar{p}$ appropriate $x$-coordinates. Otherwise, the neighborhoods $N(p)$ and $N(\bar{p})$ are the same, and we can assign $p$ and $\bar{p}$ to the opposite boundaries arbitrarily.

\else 

Each ordering of $T$ (created by ordering children at each node, or by ordering the leaves) induces a \hpolygon, of which at least one will have a visibility graph $G_{P^\prime}$ isomorphic to $G_P$. One strategy, then, is to construct $P^\prime$ for each possible ordering of $T$, check if $G_{P^\prime}$ is isomorphic to $G_P$, keeping the $P^\prime$. This would construct a polygon in time $O(n^2m + k!(n\log n+m + \exp((\log n)^{O(1)})))$ time using the quasipolynomial time graph isomorphism algorithm due to Babai~\cite{babai-2016}. However, without knowing the assignment of convex vertices to staircases, this procedure does not give us a labeled polygon.

As we now show, we can reconstruct a labeled polygon in time $O(n^2m + k!2^k(n\log n + m))$, which is faster for small values of $k$, and fixed-parameter tractable in $k$.

\begin{theorem}
Reconstructing \hpolygon s with $k$ tabs is fixed-parameter tractable, with running time $O(n^2m + k!2^{k}(n\log n + m))$.
\end{theorem}
\begin{proof}
We first construct the contact tree $T$, identify tabs and staircases in time $O(n^2m)$ following Lemma~\ref{lemma:decomposition}.  For each of the $k!$ orderings of $T$ (each of which choose a left-to-right ordering of tabs), we construct $2^k$ polygonal chains, each of which potentially form a \hpolygon. We construct each one of these polygonal chains as follows. Note that each tab $t$ sees two staircases. For each $t$ we choose one to be the up-staircase, which is to the left of $t$, and one to be the down-staircase, to the right of $t$. For each of the $2^k$ orderings of staircases, we compute the polygonal chain induced by a left-to-right ordering of tabs and staircases in $O(n+m)$ time by iterating through the staircases left to right and assigning $x$, $y$ coordinates to each \bottomV (which determines the coordinate of its \topV boundary neighbors). We then check that this polygonal chain is a \hpolygon: that is, it does not self intersect (sufficient to check for intersections with the base edge) in $O(n)$ time, and the first and last \bottomV have the same $y$ coordinate (the base edge is horizontal). If not, we stop evaluating this polygonal chain. Otherwise, call the resulting \hpolygon~$P^{\prime}$.

Let the vertices in $G_P$ be labeled from $1$ to $n$. Note that, each \bottomV (except those in the tab cliques) is assigned a vertex in $P^{\prime}$, therefore these vertices maintain their labels. We now assign the remaining vertices. Note that each \topV $v$ sees all vertices in its rectangle of the decomposition and zero or more \bottomVs outside of its rectangle. We therefore assign each \topV $v$ to be on the left (right) of its rectangle when its neighbors outside of its rectangle are to the right (left) of the left (right) side of the rectangle. This assigns $v$ the $x$ coordinate of the \topV on the left (right) side of its rectangle, respectively. If the \topVs have the same neighborhoods, then we assign them to left and right arbitrarily. If a \topV has neighbors to the left and right, then visibility is not consistent with $G_P$ and we stop evaluating.We assign \bottomVs from each tab clique similarly: the \bottomV is left (right) if it can see a bottom vertex to the right (left) the tab's staircases. This gives us a mapping of all vertices in $G(P)$ to vertices in $P^{\prime}$.

We then compute the visibility graph $G_{P^\prime}$ in time $O(n\log n + m)$ with the algorithm of Ghosh and Mount~\cite{ghosh-mount-optimal-visibility}, which maintains the labeling of the vertices in $P^\prime$. We then check in time $O(n+m)$ if $G_P$ and $G_P'$ are equal, by checking that vertices with the same label have the same labeled neighbors.
\end{proof}

\fi  

\begin{figure}[!tb]
\begin{center}
\includegraphics[scale=0.6]{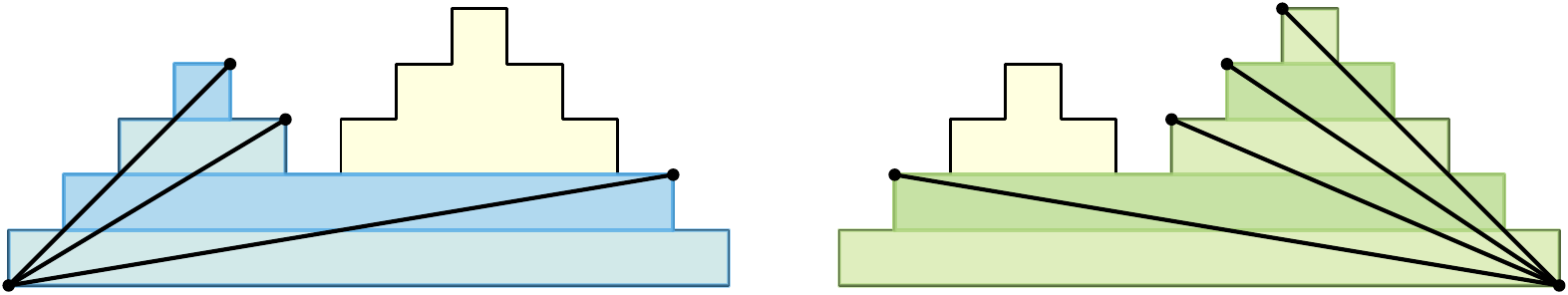}
\caption{Visibility from the left (right) base vertex determines the left- (right-)most tab, and orients all rectangles on the left (right) spine of the contact tree.}
\label{figure:leftmost-and-rightmost-paths}
\end{center}
\end{figure}

\newcommand{\vvbot}{v_{\mathrm{bot}}}
\newcommand{\vvtop}{v_{\mathrm{top}}}

\subsection{Reducing the Number of Candidate \Hpolygon s}
\label{subsection:reducing-candidates}
\iffalse
We can reduce the number of possible orderings by considering
only those that meet certain visibility constraints.
We say that a \bottomV $\vvbot$ in rectangle $R_1$ and \topV $\vvtop$ in $R_2\neq R1$ are \emph{orientation-fixed} if $\vvbot$ can see $\vvtop$.
Then these two vertices must be on opposite staircases. That is, if one vertex is on a up-staircase, the other is on a down-staircase.

Note that every \bottomV (except the \bottomV in a tab clique) is orientation-fixed with some tab vertex, as each such \bottomV sees exactly one tab vertex. And likewise each tab vertex is orientation-fixed with every \bottomV not in its tab clique. Notice then that choosing a left-to-right ordering of all the tabs fixes the staircases of all bottom vertices. 
Likewise, fixing a \bottomV $\vvbot$ fixes the staircases of all \bottomVs that are orientation-fixed with a tab vertex.
\else

We can reduce the number of possible orderings of tabs and staircases by considering only those that meet certain visibility constraints on the vertices that form the corners of each rectangle.
In particular, we say that two rectangles $R_1 \neq R_2$ in the decomposition are \emph{orientation-fixed} if a \bottomV $\vvbot$ from one can see a \topV $\vvtop$ of another. 
Then these rectangles must be oriented so that $\vvbot$ and $\vvtop$ are on opposite staircases (an up-staircase and a down-staircase). Thus, fixing an orientation of one rectangle fixes the orientation of the other.

Note that every rectangle is orientation-fixed with some leaf rectangle (as its \bottomV can see a tab vertex). Therefore, ordering (and orienting) the leaves induces an ordering/orientation of the tree. There are $O(k!2^k)$ such orderings (and orientations) for all leaf rectangles, where $k$ is the number of tabs.

For double staircases, $T$ is a path and the root rectangle is orientation-fixed with every other rectangle (a base vertex is seen by every \topV). Hence, orienting the base rectangle determines the positions of the \topVs on the double staircase.
\fi 
Likewise, for the \hpolygon, the spines of $T$ are fixed:

\begin{lemma}
The base rectangle of a \hpolygon~is orientation-fixed with all rectangles on the left and right spines of $T$.
\end{lemma}
\ifFull
\begin{proof}
The left base vertex sees all \topVs on the left-most tab's down-staircase, and the right base vertex sees all \topVs on the right-most tab's up-staircase.
\end{proof}
\fi 

Moreover, the only tab vertices visible from a base vertex are incident to the left-most or right-most tab. Thus, we can identify the left-most and right-most tabs based on the neighborhood of the base vertices.
Note that removing a base rectangle of the \hpolygon~produces one or more \hpolygon s. Then we can apply this logic recursively, leading to the following algorithm:
\begin{enumerate}
\item Fix the orientation of the base rectangle. This identifies the rectangles on the left and right spines of $T$ and their orientations. (See \fig~\ref{figure:leftmost-and-rightmost-paths}.)
\item The remaining subtrees collectively contain the remaining rectangles, which still must be ordered and oriented. We recursively compute the ordering and orientation of the rectangles in these subtrees.
\end{enumerate}

\newcommand{\children}{C}
Note if we compute the left and right spines of $T$, 
we identify the first and last tabs, and the orientations of their tab edges. Thus, we have $(k-2)!2^{k-2}$ remaining orderings of $T$ and orientations of the tab edges to check, as $k-2$ tabs remain. This results in the overall reconstruction of a \hpolygon~with $k \geq 2$ tabs in  $O(n^2m + (k-2)!2^{k-2}(n\log n + m))$ time, proving Theorem~\ref{theorem:main}.


We now generalize the number of orderings to consider by defining a recurrence on the tree structure. Let $v\in T$, and define $\children(v)$ be $v$'s children in $T$ and $d(v)=|\children(v)|$. Then if we have a fixed orientation of $v$'s corresponding rectangle, fixing the rectangles on the left-most and 
right-most paths from $v$ limits the number of possible orderings/orientations of $v$'s descendants to
%
%
\[ F(v) \leq\begin{cases}
            2^{d(v)-2}\prod_{u\in \children(v)}{F(u)} &\text{if $d(v)>1$},\\
            F(u) &\text{if $|C(v)|=1$, s.t. $C(v)=\{u\}$},\\
            1    &\text{if $v$ is a leaf}.\end{cases}
\]
Note that $F(root)=1$ for a binary tree $T$. That is, the orientation of the base rectangle completely determines the \hpolygon. Furthermore, we can find such an orientation by fixing the orientation of the base edge, determining the left- and right-most paths, ordering and orienting them to match the base edge, and then repeating this for each subtree whose root is oriented and ordered (but its children are not), which acts as a base rectangle for its subtree. This process can be done in time $O(n+m)$ by traversing $T$ and orienting each rectangle exactly once by looking at its vertices neighbors in its base rectangle in $T$.
\darren{Would be nice to say: Once again, the time is dominated by the time to compute the $1$-simplicial edges, $O(n^2m)$.}

\begin{theorem}
\Hpolygon s with a binary contact
tree can be reconstructed in $O(n^2m)$ time.
\end{theorem}

\ifFull
Note that a \hpolygon~has a binary contact tree exactly when dents at the same level in $P$ are pairwise invisible or no two dents have the same level.
\fi






\section{From Reconstruction to Recognition}
\ifFull
Note that all of our reconstruction algorithms assign each vertex to a specific position in the constructed polygon. This is not necessarily a requirement of reconstruction algorithms in general. However, as a result, the assignment enables us to develop recognition algorithms for these polygons as well. We first note the following.

\begin{theorem}
\label{theorem:reconstruction-recognition}
If there exists a vertex assignment reconstruction algorithm $\mathcal{A}$ with time $O(f(n))$ for polygons in class $\mathcal{C}$, and further if polygons in class $\mathcal{C}$ can be recognized in time $g(n)$, then the class visibility graphs of polygons in $\mathcal{C}$ can be recognized in time $O(f(n) + g(n) + n\log n + m)$. Where $n$ and $m$ are the number of vertices and edges in the input visibility graph, respectively.
\end{theorem}
\begin{proof}
Given a graph $G_P$, run $\mathcal{A}$. If $\mathcal{A}$ fails to given a vertex assignment, then $G_P$ is not in class $\mathcal{C}$. If it gives a vertex assignment, construct the visibility graph $G_{P^\prime}$ via the method of Ghosh and Mount~\cite{ghosh-mount-optimal-visibility}, which takes time $O(n\log n + m)$ with respect to the output graph. Note that the number of edges of $G_{P^\prime}$ could exceed that of $G_P$; however, we can stop the algorithm as the number of edges differ and report that $G_P$ is not in class $\mathcal{C}$. Otherwise, if the number of edges is the same, we verify that all edges are between the same vertices.
\end{proof}

Clearly, double staircases and \hpolygon s can be recognized in linear time. And since their visibility graphs have a quadratic number of edges, our reconstruction algorithms imply recognition algorithms with the same running time.
\else
We note that all of our reconstruction algorithms assign each vertex to a specific position in the constructed polygon. Let such an algorithm be called a \emph{vertex assignment reconstruction}. As a result, we get recognition algorithms for these visibility graphs as well: we run our reconstruction until it fails or completes successfully, verify that the resulting polygon has the same visibility graph in time $O(n\log n + m)$ time~\cite{ghosh-mount-optimal-visibility}, and verify that it is a polygon of the given type in linear time. Thus, we conclude that our reconstruction algorithms imply recognition algorithms with the same running times.
\fi

%


\ifFull
\newpage
\else
\ifLipics
\else
\fi 
\fi 
\bibliographystyle{abbrv}
\bibliography{orthovisibility}

\begin{thebibliography}{10}

\bibitem{abello-uniform-step-length}
J.~Abello and {\"{O}}.~E\u{g}ecio\u{g}lu.
\newblock Visibility graphs of staircase polygons with uniform step length.
\newblock {\em International Journal of Computational Geometry \&
  Applications}, 03(01):27--37, 1993.

\bibitem{abello-convex-fans}
J.~Abello, {\"{O}}.~E\u{g}ecio\u{g}lu, and K.~Kumar.
\newblock Visibility graphs of staircase polygons and the weak {B}ruhat order,
  {I}: From visibility graphs to maximal chains.
\newblock {\em Discrete Comput. Geom.}, 14(3):331--358, 1995.

\bibitem{abello-kumar-1995}
J.~Abello and K.~Kumar.
\newblock Visibility graphs of 2-spiral polygons (extended abstract).
\newblock In R.~Baeza-Yates, E.~Goles, and P.~V. Poblete, editors, {\em Proc.
  2nd Latin American Symposium}, volume 911 of {\em LNCS}, pages 1--15.
  Springer, 1995.

\bibitem{abello-kumar-2002}
J.~Abello and K.~Kumar.
\newblock Visibility graphs and oriented matroids.
\newblock {\em Discrete {\&} Computational Geometry}, 28(4):449--465, 2002.

\bibitem{asano-survey}
T.~Asano, S.~K. Ghosh, and T.~C. Shermer.
\newblock Chapter 19--visibility in the plane.
\newblock In J.-R. S.~J. Urrutia, editor, {\em Handbook of Computational
  Geometry}, pages 829--876. North-Holland, Amsterdam, 2000.

\bibitem{babai-2016}
L.~Babai.
\newblock Graph isomorphism in quasipolynomial time [extended abstract].
\newblock In {\em Proc. 48th ACM Symposium on Theory of Computing}, STOC '16,
  pages 684--697, New York, NY, USA, 2016. ACM.

\bibitem{choi-funnel-1995}
S.-H. Choi, S.~Y. Shin, and K.-Y. Chwa.
\newblock Characterizing and recognizing the visibility graph of a
  funnel-shaped polygon.
\newblock {\em Algorithmica}, 14(1):27--51, 1995.

\bibitem{colley-thesis-1991}
P.~Colley.
\newblock Visibility graphs of uni-monotone polygons.
\newblock Master's thesis, Department of Computer Science, University of
  Waterloo, Waterloo, Canada, 1991.

\bibitem{colley-cccg-1992}
P.~Colley.
\newblock Recognizing visibility graphs of unimonotone polygons.
\newblock In {\em Proc. 4th Canad. Conf. Comput. Geom.}, pages 29--34, 1992.

\bibitem{durocher-2012}
S.~Durocher and S.~Mehrabi.
\newblock Computing partitions of rectilinear polygons with minimum stabbing
  number.
\newblock In J.~Gudmundsson, J.~Mestre, and T.~Viglas, editors, {\em Proc. 18th
  Annual International Conference on Computing and Combinatorics (COCOON
  2012)}, volume 7434 of {\em LNCS}, pages 228--239. Springer, 2012.

\bibitem{elgindy-thesis}
H.~ElGindy.
\newblock {\em Hierarchical decomposition of polygons with applications}.
\newblock PhD thesis, McGill University, Montreal, Canada, 1985.

\bibitem{evens-saeedi-2015}
W.~Evans and N.~Saeedi.
\newblock On characterizing terrain visibility graphs.
\newblock {\em J. Comput. Geom.}, 6(1):108--141, 2015.

\bibitem{everett-thesis}
H.~Everett.
\newblock {\em Visibility graph recognition}.
\newblock PhD thesis, University of Toronto, 1990.

\bibitem{everett-spiral-1990}
H.~Everett and D.~Corneil.
\newblock Recognizing visibility graphs of spiral polygons.
\newblock {\em J. Algorithms}, 11(1):1--26, 1990.

\bibitem{ghosh-book}
S.~K. Ghosh.
\newblock {\em Visibility Algorithms in the Plane}.
\newblock Cambridge University Press, 2007.

\bibitem{ghosh-mount-optimal-visibility}
S.~K. Ghosh and D.~M. Mount.
\newblock An output-sensitive algorithm for computing visibility.
\newblock {\em SIAM J. Comput.}, 20(5):888--910, Oct. 1991.

\bibitem{jackson-wismath-stabs}
L.~Jackson and S.~Wismath.
\newblock Orthogonal polygon reconstruction from stabbing information.
\newblock {\em Comp. Geom.-Theor. Appl.}, 23(1):69--83, 2002.

\bibitem{orourke-art-gallery-survey}
J.~O'Rourke.
\newblock {\em Art Gallery Theorems and Algorithms}.
\newblock Oxford University Press, New York, 1987.

\end{thebibliography}

\ifFull
\else
\appendix
\newpage
\section{Omitted Proofs: Uniform-Length Orthogonally Convex Polygons}
\label{section:unit-orthogonal-appendix}

\subsection{Proof of Lemma~\ref{lemma:convex-neighbor-two-cliques}}\label{subsection:two-cliques}
To prove Lemma~\ref{lemma:convex-neighbor-two-cliques} we start with the following observations:

\begin{figure}[!tb]
\begin{center}
\includegraphics[scale=0.60]{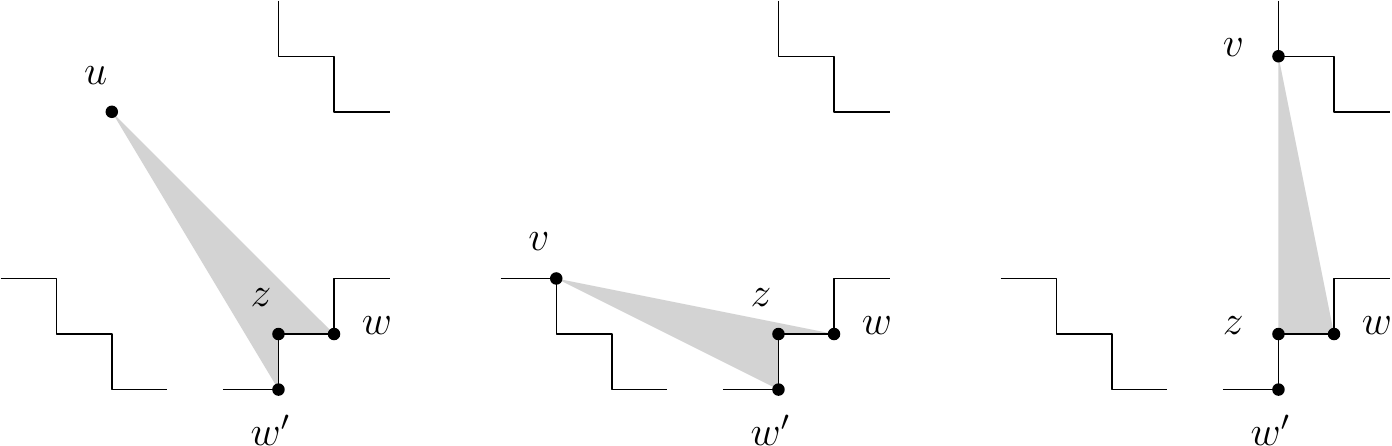}
\caption{Illustration of the proofs of Observations~\ref{obs:opposite-reflex-neighbors} (left figure) and Observation~\ref{obs:adjacent-reflex-neighbors} (middle and right) figures. Vertex $u$ on the northwest staircase can be reflex or convex.}
\label{figure:opposite-visibility}
\end{center}
\vspace*{-.5cm}
\end{figure}

\begin{observation}\label{obs:opposite-reflex-neighbors}
If a vertex $u$ on the northwest staircase of an UP sees a reflex vertex $z$ on the southeast staircase and the slope of the line supporting the line segment $\overline{uz}$ is non-positive, then $u$ sees both (convex) boundary neighbors $w$ and $w'$ of $z$.
\end{observation}

\begin{proof}
We will prove that $w$, the horizontal boundary neighbor of $z$, is visible from $u$. Proof for visibility of $w'$, the vertical boundary neighbor of $z$ is symmetric.
See \fig~\ref{figure:opposite-visibility} (left) for an illustration.

First, observe that in an UP, no boundary edge on the northeast and the southwest staircases blocks visibility between the vertices of the northwest staircase and the vertices of the southeast staircase, and, consequently, between $u$ and $w$.

Next, by viewing $u$ as the origin of the Cartesian coordinate system, $z$ lies in the southeast quadrant (inclusive of the axis), because the slope of the line supporting $\overline{uz}$ is non-positive. And since $\overline{zw}$ is a horizontal edge, $w$ also lies in the southeast quadrant. Therefore, no boundary edge of the northwest staircase blocks the visibility between $u$ and $w$ (recall that our definition of visibility allows visibility along the polygon edges).  

Finally, by viewing $w$ as the origin, we can similarly see that no edge of the southeast staircase blocks the visibility between $u$ and $w$. 
\end{proof}

\begin{observation}\label{obs:adjacent-reflex-neighbors}
Let a reflex vertex $v$ be on the southwest or northeast staircase of an UP and a reflex vertex $z$ be on the southeast staircase. If  the slope of the line supporting the line segment $\overline{vz}$ is non-positive, then $v$ sees both (convex) boundary neighbors $w$ and $w'$ of $z$.
\end{observation}

\begin{proof}
Consider the case when $v$ is on the southwest staircase (the proof of the case when $v$ is on the northeast staircase is symmetric).
See \fig~\ref{figure:opposite-visibility} (middle and right) for illustration.

First, observe that no boundary edge on the northwest and northeast staircase blocks the visibility between the vertices of the southwest staircase and the vertices of the southeast staircase, and, consequently, between $v$ and $w$.

Next, observe that the slope of the line supporting the edges from {\em any} reflex vertex on the southwest staircase and {\em any} vertex on the southeast staircase is at least $-\tan(\pi/4)$ and, therefore, no vertex on the southwest staircase can block the visibility between them.

Finally, by viewing $w$ (resp., $w'$) as the origin of the Cartesian coordinate system, we can see that $v$ is in the northwest quadrant (inclusive of axis) and, therefore, no boundary edge on the southeast staircase can block the visibility between $w$ (resp., $w'$) and $v$.
\end{proof}

\begin{figure}[!tb]
\begin{center}
\includegraphics[scale=0.60]{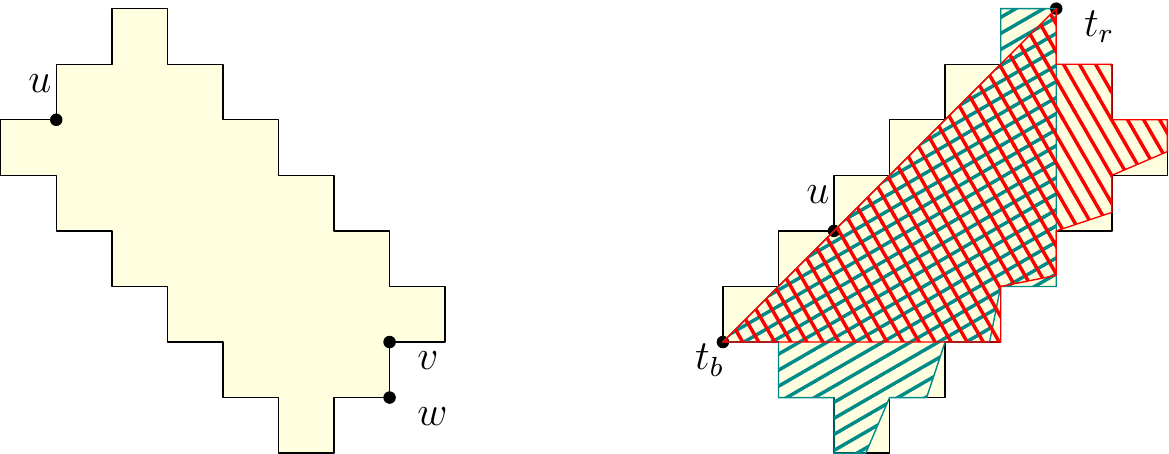}
\caption{Illustration of the proof that an edge $(u,v)$ is non-simplicial, when both $u$ and $v$ are reflex. The proof demonstrates the existence of a convex vertex $w$ visible from both $u$ and $v$.}
\label{figure:reflex-reflex}
\end{center}
\vspace*{-.5cm}
\end{figure}

We are ready to prove Lemma~\ref{lemma:convex-neighbor-two-cliques}.\\

\begin{repeatlemma}{lemma:convex-neighbor-two-cliques}
In a UP, if $u$ or $v$ is a reflex vertex, then edge $(u,v)$ is not $1$-simplicial. 
\end{repeatlemma}

\begin{proof}
Let $(u,v)\in E_P$ and suppose that at least one of $u$ and $v$ is reflex. 

\emph{Case 1: Both $u$ and $v$ are reflex.} Then they belong to a maximal clique consisting of all reflex vertices and no convex vertices. We will show that both $u$ and $v$ also see some convex vertex $w$, therefore, $u,v$ and $w$ are part of another maximal clique. For concreteness of exposition, we orient the polygon so $u$ is always on the northwest staircase. There are two cases to consider:

\begin{enumerate}[(a)]
\item {\bf Both $u$ and $v$ are on the short staircases.} (See \fig~\ref{figure:reflex-reflex} (left)). Observe, that every reflex vertex on a short staircase sees all vertices on the other short staircase. Then, if $u$ and $v$ are on the same (northeast) staircase, there is a convex vertex $w$ on the southeast staircase, visible to both $u$ and $v$. If $v$ is on the southwest staircase, $w$ is either of the two convex boundary neighbors of $v$. Since $w$ is a boundary neighbor of $v$, clearly it is visible from $v$. 

\item {\bf At least one of $u$ and $v$ is on a long staircase.} (See \fig~\ref{figure:reflex-reflex} (right)). Without loss of generality, let it be $u$, which is on the northwest staircase. Vertex $u$ sees the right tab vertex $t_r$ of the north tab and the bottom tab vertex $t_b$ of the west tab. Since northwest staircase is long, the union of reflex vertices visible from $t_r$ and $t_b$ is the set of all reflex vertices in the polygon. Therefore, every reflex vertex $v$ must see at least one of $t_r$ or $t_b$. 
\end{enumerate}

\begin{figure}[!tb]
\begin{center}
\includegraphics[scale=0.60]{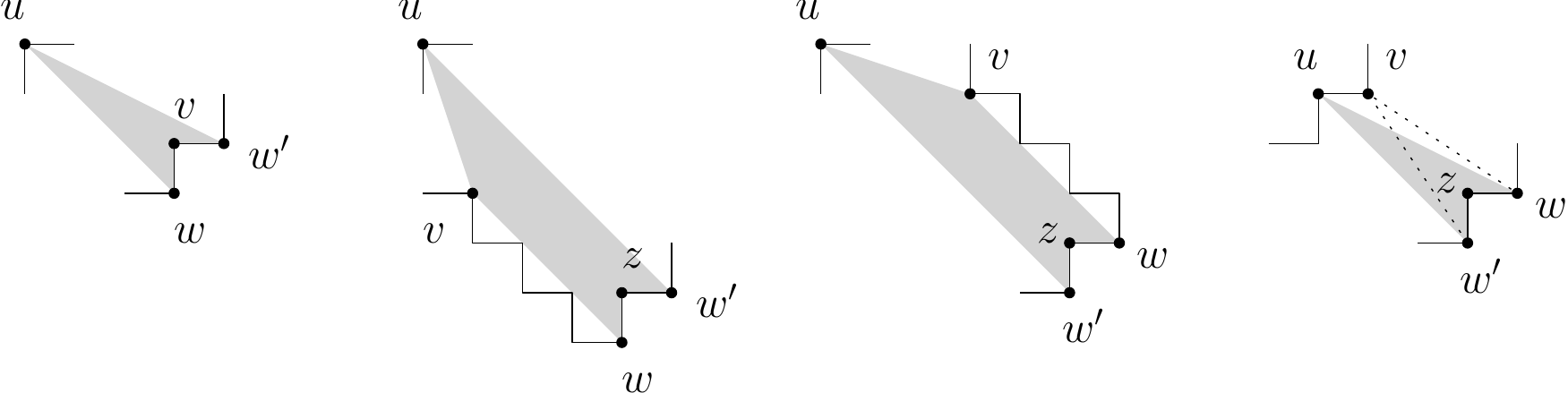}
\caption{Illustration of the proof that an edge $(u,v)$ is non-simplicial, when one vertex, $u$, is convex and another, $v$, is reflex. For each of the four cases of when $v$ is on the four different staircases, there is a pair of convex vertices $w$ and $w'$ on the southeast staircase, which are visible to both $u$ and $v$, but not to each other.}
\label{figure:convex-reflex}
\end{center}
\vspace*{-.5cm}
\end{figure}
\emph{Case 2: One of the vertices is convex.} Without loss of generality, let it be $u$ (again, on the northwest staircase), and let $v$ be reflex. Let $P_u$ be a polygon, which consists of $u$ and all reflex vertices visible from $u$. For every reflex  vertex $v$ in $P_u$, we will identify two convex vertices $w$ and $w'$ on the southeast staircase, which are visible to both $u$ and $v$. Since $w$ and $w'$ are convex vertices on the same staircase, they are invisible to each other and, therefore, $u,v,w$ and $u,v,w'$ will be part of two distinct maximal cliques. 

Let us consider  the four possible locations for $v$  (see \fig~\ref{figure:convex-reflex} for illustrations):

\begin{enumerate}[(a)]
\item {\bf $v$ is on the southeast staircase.} Let $w$ and $w'$ be the two convex boundary neighbors of $v$. Clearly, $v$ sees $w$ and $w'$. Moreover, since $u$ sees $v$ and $v$ is on the opposite staircase from $u$, the slope of the line supporting  the segment $\overline{uv}$ is non-positive. Then, by Observation~\ref{obs:opposite-reflex-neighbors}, $u$ must also see $w$ and $w'$. 

\item \label{item:southwest}{\bf $v$ is on the southwest staircase.} Let $z$ be the lowest reflex vertex on the southeast staircase (adjacent to the south tab's right vertex) and let $w$ and $w'$ be the two convex boundary neighbors of $z$.  Vertex $z$ is visible from $v$ because both $v$ and $z$ are reflex. Since $z$ is the lowest reflex vertex, the slope of the line supporting the segment $\overline{vz}$ is non-positive (the slope is 0 if $v$ is the lowest reflex vertex on the southwest staircase). Thus, by Observation~\ref{obs:adjacent-reflex-neighbors}, $w$ and $w'$ are visible from $v$. Moreover, since $u$ sees $v$ and $v$ is on the southwest staircase, $z$ must be visible from $u$ with even smaller slope of the line supporting segment $\overline{uz}$, i.e., the slope is also non-positive. Therefore, by Observation~\ref{obs:opposite-reflex-neighbors}, $w$ and $w'$ are also visible from $u$.

%

\item {\bf $v$ is on the northeast staircase.} Let $z$ be the highest reflex vertex on the southeast staircase (adjacent to the east tab's bottom vertex) and let $w$ and $w'$ be the two convex boundary neighbors of $z$. Again, vertex $z$ is visible from $v$ because both $v$ and $z$ are reflex. We can see that the slope of the line supporting the segment $\overline{vz}$ is negative by viewing $v$ as the origin of the Cartesian coordinate system and observing that $z$ is in the southeast quadrant. Thus, by Observation~\ref{obs:adjacent-reflex-neighbors}, $w$ and $w'$ are visible from $v$. At the same time, $z$ is visible from $u$ because $v$ is visible from $u$ and $z$ is below and to the right of (or directly below) $v$. Moreover, the slope of the line supporting segment $uz$ is also non-positive. Therefore, by Observation~\ref{obs:opposite-reflex-neighbors}, $w$ and $w'$ are also visible from $u$.


\item {\bf $v$ is on the northwest staircase.} Since $v$ is visible from $u$,  $v$ must a reflex boundary neighbor of $u$. Observe, that because $P_u$ is bounded from all sides, it must contains a non-empty subset $Z$ of reflex vertices from the southeast staircase. Moreover, there exists a vertex $z\in Z$, which is not axis-aligned with $u$, i.e., the slope of the line supporting the segment $\overline{uz}$ is strictly negative. Let $w$ and $w'$ be the two convex boundary neighbors of $z$. By definition of $P_u$, $z$ is visible from $u$, therefore, by Observation~\ref{obs:opposite-reflex-neighbors}, $w$ and $w'$ are visible from $u$. Again, $z$ is visible from $v$ because both $v$ and $z$ are reflex. Since the slope of the line supporting segment $\overline{uz}$ is strictly negative and the polygon has uniform length boundaries, the slope of the line supporting segment $\overline{vz}$ must be non-positive. 
Therefore, by Observation~\ref{obs:opposite-reflex-neighbors}, $w$ and $w'$ are also visible from $v$.
\end{enumerate}
\end{proof}

\subsection{Proof of Lemma~\ref{lemma:iup-tabs}}
\label{subsection:find-4-tabs}

\begin{repeatlemma}{lemma:iup-tabs}
We can identify the four tabs of an IUP in $O(nm)$ time.
\end{repeatlemma}

\begin{proof}
We compute the four $7$-vertex maximal cliques of Lemma~\ref{lemma:cliques-convex-vertices} in $O(nm)$ time using Observation~\ref{claim:simplicial-k-clique-time}.\nodari{Can't we do it faster, e.g. in $n^2$ time? We should state the tightest bound here and update Lemma~\ref{lemma:iup-tabs}}\darren{We must evaluate each of the $m$ edges, and spend $O(n)$ time on each edge determining if the edge is $1$-simplicial and in a maximal $7$-clique.} These cliques have exactly three convex vertices each, and tabs are incident to two convex vertices, narrowing our choice of tab down to $4\cdot\binom{3}{2} = 12$ edges. Four of these are vertical or horizontal (non-boundary) edges, which we can detect and eliminate, as they are not $1$-simplicial: these edges are in both a clique with both tab vertices, and a clique containing a tab vertex, another convex vertex, and reflex vertices not in the tab clique. We have eight remaining edges to consider.  
These eight edges form four disjoint paths on two edges, and the middle vertex on each path is a tab vertex. 

Note that these middle tab vertices are on the long staircases. Let one of them be called $u$. Now it remains to find $u$'s neighbor on its tab. Vertex $u$ has two candidate neighbors; let's call them $v$ and $w$. Just for concreteness, let's say $u$ is the vertex of the north tab on the northeast staircase.

Suppose, without loss of generality, that $v$ has more reflex neighbors than $w$, then $v$ is $u$'s neighbor on a tab, because it sees reflex vertices on the whole northeast and southeast staircases, while $w$ sees only a subset of those.
Otherwise $v$ and $w$ have the same number of reflex neighbors, which only happens when $w$ sees every reflex vertex on the southeast and northeast staircases. Then either $v$ or $w$ has more convex neighbors. Suppose, without loss of generality, that $v$ is a tab vertex, then $v$ has fewer convex neighbors than $w$. To see why, note that since $u$ is on northeast (long) staircase, $v$ is on the northwest (short) staircase. Vertex $v$ has convex neighbors $u$, $w$, and every convex vertex on the southeast (short) staircase. Likewise, $w$ has convex neighbors $v$, $u$, every convex vertex on the northeast (long) staircase (including $u$) and one vertex on the southeast (short) staircase. 
\nodari{This proof lacks a smooth flow... Reads like a list of observations, without overall intuition guiding the reader. I am going to rewrite it later.}

We can do these checks for all such pairs $v$ and $w$, giving us all tabs.
\darren{Brief time breakdown.} Note that, if we have already distinguished convex and reflex vertices, this step takes time $O(nm)$: $O(nm)$ time to compute the four cliques containing tabs, and $O(n)$ to count the number of reflex and convex vertices visible from each tab vertex candidate.
\end{proof}

\subsection{Proof of Lemma~\ref{lemma:elementary-cliques}}
\label{subsection:elementary-cliques}

\begin{repeatlemma}{lemma:elementary-cliques}
We can identify the elementary cliques containing vertices on the northwest staircase in $O(nm)$ time.
\end{repeatlemma}
\begin{proof}
Each elementary clique $C$ contains a $1$-simplicial edge, as $C$ is maximal and two convex vertices in $C$ must be on opposite staircases. Using Observation~\ref{claim:simplicial-k-clique-time} we compute all $1$-simplicial edges in maximal cliques on seven or nine vertices in $O(nm)$ time, and keep the cliques that are elementary cliques---those that contain three convex and either four or six reflex vertices.

Let $k_{NW}$ be the number of convex vertices on the northwest staircase. We number these convex vertices from $v_0$ to $v_{k_{NW}-1}$ in order along the northwest staircase from the north tab to the west tab. We denote by $C_i$ the unique elementary clique containing $v_i$. Note that $C_0$ is the unique maximal (elementary) clique containing the north tab. Furthermore, each clique $C_i$ contains a set of reflex vertices $R_i$ such that $|R_i\cap C_{i+1}| = 3$, and for $j \not\in\{i-1,i,i+1\}$, $R_i\cap C_j = \emptyset$.\darren{Expand to explain why empty.}
Therefore, from elementary clique $C_i$, we can compute elementary clique $C_{i+1}$ by searching for the only other elementary clique containing reflex vertices $R_i$. Once we reach an elementary clique containing a tab, then we have computed all elementary cliques on the northwest staircase. This tab is the west tab and we are finished.
\end{proof}

\subsection{Proof of Lemma~\ref{lemma:long-staircases}}
\label{subsection:long-staircases}
\newcommand{\W}{\mathcal{W}}
\newcommand{\E}{\mathcal{E}}

\begin{repeatlemma}{lemma:long-staircases}
We can assign the remaining convex vertices 
in $O(n^2)$ time.
\end{repeatlemma}

\begin{proof}
Let $\W$ and $\E$ be all the convex vertices on the southwest and northeast staircases (which we are computing) and let $\W_0$ and $\E_0$ be the convex vertices  on the southwest and northeast staircases that are already known from the elementary cliques from Lemma~\ref{lemma:cliques}. Let $N_c(v)$ denote the set of convex neighbors on the opposite staircase of some vertex $v$. Then, for each vertex $w_0\in \W_0$, $N_c(w_0) \subseteq \E$, i.e., the convex neighbors of the (convex) vertices in $\W_0$ are on the northeast staircase. Similarly, for each vertex $e_0\in \E_0$, $N_c(e_0)\subseteq \W$. Then we can iteratively define sets $\E_i = (\cup_{w\in \W_{i-1}} N_c(w)) \setminus \E_{i-1}$ and $\W_i =(\cup_{e\in \E_{i-1}} N_c(e)) \setminus \W_{i-1}$ and identify all vertices of the southwest and northeast staircases as $\W = \cup_i \W_i$ and $\E = \cup_i \E_i$. 

To order the vertices along the southwest staircase, note that the sets $\W_i$ should appear in order of increasing $i$ from top to bottom. Also note that if one were to assign the vertices of $\W_i$ to a staircase from top to bottom, each vertex $w_i$ in this order would see fewer vertices of $\E_{i-1}$. Thus, we can order the vertices within each $\W_i$. The argument for ordering vertices of $\E_i$ is symmetric.

Computing set $\E_i$ takes time $O(|\W_{i-1}|\cdot n)$: we evaluate all $O(n)$ neighbors of each neighbor in $\W_{i-1}$. Likewise, computing $\W_i$ takes time $O(|\E_{i-1}|\cdot n)$. Thus, overall it takes $O((|\W| + |\E|)\cdot n) = O(n^2)$ time to construct (and order) sets $\W$ and $\E$.
\end{proof}

\subsection{Proof of Lemma~\ref{lemma:remaining-reflex}}
\label{subsection:remaining-reflex}

\begin{repeatlemma}{lemma:remaining-reflex}
We can assign the reflex vertices to each staircase in $O(n^2)$ time.
\end{repeatlemma}
\begin{proof}
Once the convex vertices are ordered on the staircases, we can compare the reflex vertices that are seen from the tab vertices. Let $a$ and $b$ be vertices on different tabs visible along a short staircase \ifFull(see \fig~\ref{figure:side-visibility}(left))\else(see \fig~\ref{figure:algorithm-stages}\subref{ab-visibility})\fi. Let $\R$ be the set of all reflex vertices of the IUP, and $N(v)$ be a set of all neighbors of vertex $v$ in the visibility graph of IUP. Then $\R_0 = N(a)\cap N(b)\cap \R$ contains all reflex vertices from the short staircase, plus two extra reflex vertices from its neighboring long staircases. The remaining vertices $N(a)\setminus \R_0$ are on one long staircase (and $N(b)\setminus R_0$ are on the other long staircase).

Thus, we can find many reflex vertices on the long staircases in $O(n)$ time, except the endvertices and potentially those in the middle of the staircases. To find the remaining ones, we build rectangles (maximal cliques) consisting of two convex vertices $u$ and $v$ on the opposite staircases and a known reflex vertex $w$, such that $(u,w)$ forms a boundary edge of the IUP\nodari{how do we know which edges in $G_P$ are boundary edges in $P$?} \ifFull(see \fig~\ref{figure:side-visibility}(right))\else(see \fig~\ref{figure:algorithm-stages}\subref{rectangles})\fi. These rectangles define new reflex vertices on the staircase opposite from $w$. Thus, we iteratively discover all new reflex vertices. 

Given the two convex vertices of the rectangle, it takes $O(n)$ to compute the maximal $6$-clique of the rectangle, and $O(1)$ to determine the new reflex vertex. We must do this for $O(n)$ rectangles. Hence, the total time to discover these new reflex vertices is $O(n^2)$. 
\end{proof}

\section{Omitted Proofs: \Hpolygon s}
\label{section:monotone-staircase-appendix}

\subsection{Reconstructing a Double Staircase Polygon in Linear Time}
\label{subsection:double-staircase}
Every double staircase polygon can be decomposed into $k$ axis-aligned rectangles for some integer $k \ge 1$. Note that the number of vertices in such polygon is $4k$, with each vertex $u$ on one of $k+1$ levels $l_u = 0, \dots, k$. With the exception of levels $l_u = 0$ and $l_u = k$, which contain only two vertices (pairs of tab and base vertices), there are four vertices per level (two top and two bottom vertices).

Observe that the degree $deg(u)$ of each vertex $u \in G_P$ exhibits the following pattern:

\begin{enumerate}[(a)]
\item {\bf $u_k$ is a tab vertex (on level $k$).} $deg(u_k) = k + 2$: $u_k$ is a neighbor to $k$ bottom vertices on the opposite staircase and two boundary vertices of $u$.
\item {\bf $u_l$ is a top vertex on level $l = 1,\dots, k-1$.} $deg(u_l) = l+4$: $u_l$ is a neighbor to $l+1$ bottom vertices on the opposite staircase, 1 convex vertex on the same level as $u_l$ on the opposite staircase, and two boundary vertices of $u_l$. 
\item {\bf $u_l$ is a bottom vertex on level $l = 1,\dots, k-1$.} $deg(u_l) = 3k-l + 2$: $u_l$ is a neighbor to $k$ bottom vertices on the opposite staircase, $k- l+1$ top vertices on the opposite staircase, $k-1$ bottom vertices on the same staircase and two boundary vertices of $u_l$. 
\item {\bf $u_0$ is a base vertex (on level $0$).} $deg(u_0) = 3k$: $u_0$ is a neighbor to $2k$ vertices on the opposite staircase, $k-1$ bottom vertices on the same staircase, and one boundary vertex of $u_0$.
\end{enumerate}

Observe, that most of the above values for degrees come in pairs: one vertex per staircase. For positive integer $k$, the only two exceptions are $deg(u) = k+2$ and $deg(u) = 3k$, each of which appears four times: two tab vertices $u_k$ and $u'_k$ and two top vertices $u_{k-2}$ and $u'_{k-2}$ (on level $k-2$), and two base vertices $u_0$ and $u'_0$ and two bottom vertices $u_2$ and $u'_2$ (on level $2$). However, we can easily differentiate tab vertices $u_k$ and $u'_k$ from top vertices $u_{k-2}$ and $u'_{k-2}$, because $u_k$ and $u'_k$ are neighbors to the two bottom vertices on level $k-1$ (of degree $2k+3$), while $u_{k-2}$ and $u'_{k-2}$ are not. Similarly, we can differentiate the two base vertices $u_0$ and $u'_0$ from bottom vertices $u_2$ and $u'_2$, because $u_0$ and $u'_0$ are neighbors to the two top vertices on level one (of degree $5$), while $u_2$ and $u'_2$ are not.

Thus, we can identify the levels of each vertex by their degrees (or the degrees of their neighbors). Finally, an arbitrary left-right assignment of the two base vertices to the two staircases, defines the assignment of all top (convex) vertices (including tab vertices) to the staircase. The assignment of tab vertices to the staircases defines  the assignment of all reflex vertices to staircases. Such assignment can be performed in linear time using the classification of the vertices into top and bottom vertices based on their degrees. 


\subsection{Proof of Lemma~\ref{lemma:isolated-vertices}}
\begin{figure}[!tb]
\begin{center}
\includegraphics[width=0.9\textwidth]{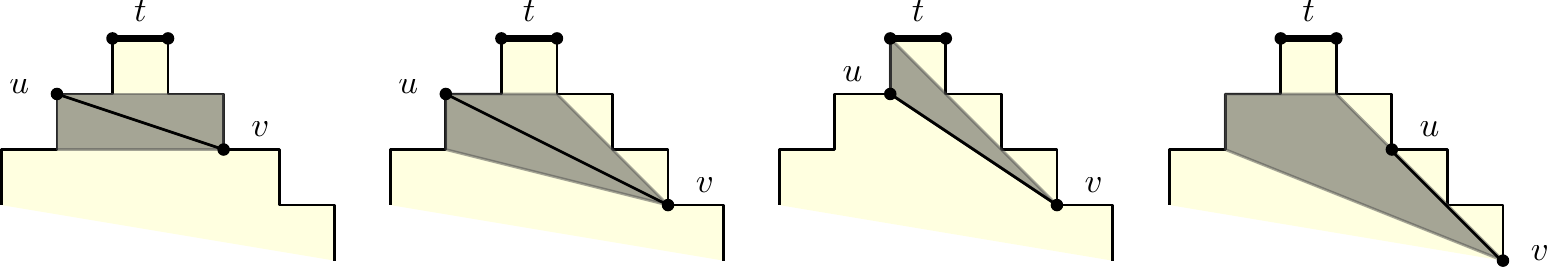}
\caption{Edges between $u,v$ on a tab $t$'s staircases, but disjoint from $t$, are in a clique of size at least $5$.}
\label{figure:cliques-in-tab-staircases}
\end{center}
\vspace*{-.5cm}
\end{figure}
\label{subsection:islated-vertices}

\begin{repeatlemma}{lemma:isolated-vertices}
In a \hpolygon, every $1$-simplicial edge in a maximal $4$-clique contains either a tab vertex or an \iso. 
\end{repeatlemma}
\begin{proof}
Let $\ell_u$ denote the \emph{level} of vertex $u$ in $P$, which is its $y$-coordinate assuming all uniform edges have unit length, where the base vertices are at level 0. 
Consider an arbitrary $1$-simplicial edge $(u,v)$ that is part of a maximal $4$-clique $C$. We assume that neither $u$ nor $v$ is a tab vertex, otherwise we are done. Note that if $|\ell_u - \ell_v| \le 1$ then $u$ and $v$ are in an axis-aligned rectangle in $P$ defined by at least six vertices of $G_P$ and thus $|C| \ge 6$. Thus, suppose that $\ell_u - \ell_v \ge 2$ and, without loss of generality, $u$ lies to the right of $v$ (the case of $u$ lying to the left of $v$ can be proven symmetrically). Since a \topV does not see any vertices above it, $v$ must be a \bottomV. Thus, $v$ sees a vertex of some tab $t$. We will show that $t$ cannot see any other vertex of $C$. Let $R$ be a set of reflex vertices on $v$'s staircase on levels $\ell_{v}+1$ up to and including $\ell_u$. Observe, that $v$ cannot be part of the maximal $4$-clique that contains both vertices of $t$, hence, $1 \le |R| \le \ell_u - \ell_v$ ($|R|< \ell_u-\ell_v$ when the vertices of $t$ are below $u$).

\emph{Case 1:} $u$ and $v$ are vertices of an up- and down-staircase of $t$. (See \fig~\ref{figure:cliques-in-tab-staircases}.)
\begin{enumerate}
\item {\bf $u$ is a \topV:} $u$ and $v$ must lie on different staircases, and $(u,v)$ is in a clique consisting of $u$, $(\ell_u-\ell_v) $ reflex vertices on $v$'s staircase, $v$, and $u$'s two (reflex) boundary neighbors and thus $|C|\geq 6$.

\item {\bf $u$ is a \bottomV:} Then $(u,v)$ is in a convex quadrilateral (clique) consisting of at least five vertices: $(\ell_u-\ell_v+1) \geq 3$ \bottomVs on $v$'s (up-)staircase from levels $\ell_v$ through $\ell_u$, and on the opposite (down-)staircase: a \bottomV on level $\ell_u$ and \topV on level $\ell_u+1$.
\end{enumerate}

\emph{Case 2:} Vertices $u$ and $v$ belong to up- and down-staircases of different tabs. Then we call $(u,v)$ a \emph{crossing edge}. Consider the following cases:
\begin{enumerate}
  \item {\bf $u$ is reflex:} Let $(u, u')$ be a horizontal boundary edge. 
    \begin{enumerate}
      \item {\bf $u'$ is convex:} Let $(u', u'')$ be a vertical boundary edge. Then $v$ sees $u''$, $C = \{v, u, u', u''\}$ and $t$ does not see $u, u',$ or $u''$.
      \item {\bf $u'$ is reflex:}   Either some vertex in $R$ sees $u$ and $u'$ (and, consequently, is in $C$) or there is a vertex $w \in C$, such that line segments $\overline{vw}$ and $\overline{wu}$ define the boundary of the convex region of $C$ which exclude the vertices of $R$.  At the same time, there must be at least one vertex $w' \in C$, bounding the convex region of $C$ on the right (e.g. by line segments $\overline{u'w'}$ and $\overline{w'v}$. Either way, $|C|  > 4$.
    \end{enumerate}
  \item {\bf $u$ is convex:}  Let $(u, u')$ and $(u'', u)$ be the vertical and horizontal boundary edges, respectively. Since $v$ sees $u$ and $u'$ is below $u$, $v$ must see $u'$, i.e. $u' \in C$, but $t$ does not see $u'$. 
    \begin{enumerate}
      \item  {\bf $v$ does not see $u''$:} There must be a reflex vertex $w \in C$, that blocks $v$ from seeing $u''$. Note that $u$ sees both $v$ and $u''$ and consequently cannot belong to $v$'s double staircase, i.e. $t$ does not see $w$. Thus $C = \{v, u, u', w\}$ and $t$ sees $v$, but not $u$, $u'$, or $w$.

      \item {\bf $v$ sees $u''$:} In this case, either some vertices of $R$ are in $C$ or there is some other vertex $w \in C$ blocking them from $u$, $u'$ or $u''$. In either case, since $\{v, u, u', u''\} \subsetneq C$, $|C| > 4$.  
    \end{enumerate}%
\end{enumerate} %
%
%
%
%
%
%
%
%
\end{proof}

\subsection{Proof of Lemma~\ref{lemma:monotone-tab-id}}
\label{subsection:monotone-tab-id}
\begin{repeatlemma}{lemma:monotone-tab-id}
In a visibility graph of a \hpolygon, tabs can be computed in time $O(n^2m)$.
\end{repeatlemma}
\begin{proof}
See \fig~\ref{figure:monotone-staircase-tab-candidates}.
We begin by computing all $1$-simplicial edges in maximal $4$-cliques, which takes time $O(n^2m)$ by Observation~\ref{claim:simplicial-k-clique-time}. Call this set of edges $\esim$, and the set of their maximal cliques $\csim$. Then $\esim$ contains the tabs, some edges that share a vertex with the tabs, and edges between staircases of different tabs (crossing edges) (which contain {\isos} by Lemma~\ref{lemma:isolated-vertices}). For all (non-incident) pairs of $1$-simplicial edges $e_1$ and $e_2$ in maximal $4$-cliques $C_1$ and $C_2$, respectively, we check if exactly one vertex of $C_2$ can be seen by an endvertex of $e_1$. That is, we compute the set $\test =\{v\in C_2\mid (u,v)\in E \text{ and } u\in e_1\}$ and verify that $|\test|=1$. If $e_1$ is a tab, then $C_2$ contains an \iso, and is detected as a non-tab clique. Thus, if we compare all pairs of edges and cliques, all $4$-cliques containing crossing edges will be eliminated. We can do this check in $O(nm)$ time by first storing, for each vertex $u$, the edges $\{(u,v)\in \esim\}$ and cliques $\{C\in\csim\mid u\in C\}$. Then for each edge $(u,v)$ in $G_P$, we run the {\iso} check for each pair of edges and cliques stored at the endvertices $u$ and $v$. Each check of all pairs takes $O(n^2)$, and we do this for $|\esim| = O(m)$ edges.

If only disjoint cliques remain after the previous step, then we have computed exactly all $k$ tabs. Otherwise, we need to eliminate non-tab edges that share a tab vertex. Note that non-tab edges form cliques along the staircase incident to the tab. Since staircases in a \hpolygon~are disjoint, 
non-tab cliques only intersect where they intersect a tab clique. Therefore, the remaining set of cliques can be split into $k$ mutually disjoint sets, where $k$ is the number of top tabs, and each set has at most three cliques that intersect (see \fig~\ref{figure:monotone-staircase-tab-candidates}), of which exactly one is a tab clique, and the other at most two cliques contain a tab vertex, and its non-tab neighbors on the opposite staircase.
Let $S$ be one of these $k$ sets. We can compute the tab in $S$ as follows (three cases):
\begin{enumerate}
\item ($|S|=1$) Then the tab vertices see fewer vertices than non-tab (reflex) vertices. We can see this as follows: tab vertices see the vertices in their tab clique, plus the \bottomVs on the tab's staircases. The non-tab reflex vertices see the same vertices, plus at least two other vertices at their same level, which tab vertices cannot see.
\item ($|S|=2$) The two cliques in $S$ share three vertices, and two vertices $u$ and $v$ that are in exactly one unique clique each. One of these vertices (say $u$) is on the tab, and sees fewer vertices than the other non-tab vertex ($v$) does. We can see this as follows: assume without loss of generality that $u$ is on the left of its double staircase. Then $u$ sees the three vertices of its tab clique, and all reflex vertices on the left staircase. Meanwhile $v$ is on the same side (left) as $u$, and sees the same vertices as $u$, plus a (convex) boundary neighbor, which is on the level between $u$ and $v$, which $u$ cannot see. The remaining tab vertex is adjacent to $v$ along its $1$-simplicial edge forming the clique in $S$.
\item ($|S|=3$) The cliques intersect in a symmetric pattern. The tab edge is formed between the two vertices that are in exactly two of these maximal cliques.
\end{enumerate}
There are at most $3k = O(n)$ of these overlapping cliques in total, and they can be separated into their respective $k$ disjoint sets in time $O(k)$ by marking the vertices of each set, and collecting the intersecting sets. Then within each set, it takes $O(1)$ time to find the tab. Thus the running time is dominated by the time to detect crossing edges in $\esim$: $O(n^2m)$.
\end{proof}

\fi{}

\end{document}